\documentclass[
 reprint,
 aps,
pra,
letterpaper,
longbibliography,floatfix]{revtex4-2}
\usepackage{natbib}
\usepackage{amsthm, amsmath,amssymb}
\usepackage{graphicx}
\usepackage{dcolumn}
\usepackage{hyperref}

\usepackage{ragged2e}
\hypersetup{
    colorlinks=true,
    linkcolor=blue,
    filecolor=blue, 
    citecolor=blue,
    urlcolor=blue,
}
\usepackage{braket}
\usepackage{bm}

\usepackage{multirow}

\newcommand{\tr}{\mathrm{Tr}}

\newcommand{\cC}{\mathcal{C}}

\newcommand{\cE}{\mathcal{E}}

\newcommand{\cH}{\mathcal{H}}

\newcommand{\cK}{\mathcal{K}}

\newcommand{\cM}{\mathcal{M}}
\newcommand{\cN}{\mathcal{N}}
\newcommand{\cO}{\mathcal{O}}
\newcommand{\cP}{\mathcal{P}}

\newcommand{\cR}{\mathcal{R}}
\newcommand{\cS}{\mathcal{S}}

\newcommand{\cV}{\mathcal{V}}

\newtheorem{theorem}{Theorem}
\newtheorem{corollary}{Corollary}[theorem]
\newtheorem{lemma}{Lemma}

\begin{document}

\title{Noise-adapted recovery circuits for quantum error correction}

\author{Debjyoti Biswas}

\affiliation{Department of Physics, IIT Madras, Chennai - 600036.}

\author{Gaurav M. Vaidya}

\affiliation{Department of Physics, IIT Madras, Chennai - 600036.}

\author{Prabha Mandayam}
\affiliation{Department of Physics, IIT Madras, Chennai - 600036.}

\date{\today}

\begin{abstract}
Implementing quantum error correction (QEC) protocols is a challenging task in today's era of noisy intermediate-scale quantum devices. We present quantum circuits for a universal, noise-adapted recovery map, often referred to as the Petz map, which is known to achieve close-to-optimal fidelity for arbitrary codes and noise channels. While two of our circuit constructions draw upon algebraic techniques such as isometric extension and block encoding, the third approach breaks down the recovery map into a sequence of two-outcome POVMs. In each of the three cases we improve upon the resource requirements that currently exist in the literature. Apart from Petz recovery circuits, we also present circuits that can directly estimate the fidelity between the encoded state and the recovered state. As a concrete example of our circuit constructions, we implement Petz recovery circuits corresponding to the $4$-qubit QEC code tailored to protect against amplitude-damping noise. The efficacy of our noise-adapted recovery circuits is then demonstrated through ideal and noisy simulations.

\end{abstract}

\maketitle

\section{Introduction \label{sec:intro}}

Quantum states are inherently fragile and error-prone, rendering the current generation of quantum processors extremely noisy~\cite{preskill2018}. Quantum error correction (QEC)~\cite{terhal_qec} provides a means by which quantum information maybe protected from noise, paving the way for robust and scaleable quantum processors. The standard approach to QEC, which relies on general-purpose codes that correct for \emph{arbitrary} single-qubit noise, is however rather resource-intensive. The smallest surface codes that correct for arbitrary single-qubit errors, for example, require at least $13$ physical qubits for a single round of QEC~\cite{tomita2014}. 

Noise-adapted QEC, on the other hand, aims to identify the dominant noise affecting the quantum system and tailor the quantum code and recovery to correct for that specific noise model. It has been shown~\cite{leung, hkn_pm2010, jayashankar2020finding} that such noise-adapted QEC schemes often require fewer resources while achieving comparable levels of error mitigation as general-purpose QEC. While there have been several analytical and numerical studies to identify noise-adapted quantum codes and adaptive recovery maps for specific noise models (see \cite{jayashankar2022} for a recent review), proposals to implement such noise-adapted QEC schemes have been few and far in between.

One key challenge in realising noise-adapted QEC is that of implementing the recovery map associated with such schemes. Since standard QEC works by breaking down arbitrary noise in terms of Pauli errors, the recovery operation that follows syndrome detection is simply a Pauli frame change~\cite{nielsen}. Adaptive QEC, on the other hand, does not view the noise in terms of a Pauli error basis and hence typically requires non-trivial recovery circuits corresponding to different error syndromes~\cite{fletcher2008, jayashankar_ft}. Constructing such non-Pauli recovery circuits is often hard -- indeed, there are very few explicit adaptive recovery schemes discussed in the literature, even for the well-studied case of amplitude-damping noise. 

An important universal prescription for noise-adapted recovery is the so-called \emph{Petz map}~\cite{petz2003, barnum2002}. This is a completely positive trace-preserving (CPTP) map which has served as an important analytical tool in the context of approximate QEC~\cite{hkn_pm2010, mandayam2012} and noise-adapted QEC~\cite{jayashankar2020finding}. Variants of the Petz map also play an important role in quantum Shannon theory~\cite{junge2018} and more recently, in the context of operator reconstruction in the AdS/CFT correspondence~\cite{chen2020}. 

A quantum algorithm to implement the Petz map was proposed recently~\cite{gilyen2022_petz}, providing, for the first time, a systematic procedure for circuit realizations of the map. However, this algorithm is geared towards implementing the \emph{state-specific} Petz map~\cite{barnum2002} rather than the \emph{code-specific} Petz map~\cite{hkn_pm2010}. The latter, which can be thought of as a special case of the state-specific Petz map, is arguably of greater relevance in the context of quantum error correction since it has been shown to correct with near-optimal fidelity for the action of a given noise channel $\mathcal{E}$ on a quantum code $\cC$~\cite{hkn_pm2010}.

In the present work, we demonstrate three different circuit constructions for implementing the code-specific Petz map and estimate the resource requirements in each case. Two of these approaches are based on known algebraic techniques for implementing quantum dynamical maps~\cite{lloyd2001engineering, terashima2005nonunitary}. The third approach combines newly developed algorithmic techniques such as block encoding~\cite{Gily_n_2019} with the well-known isometric extension. As a concrete use case, we obtain Petz recovery circuits for the specific case of a $4$-qubit code~\cite{leung} tailored to protect against single-qubit amplitude-damping noise. Finally, we simulate these circuits on noisy superconducting processors available on the IBMQ platform and benchmark their performance in terms of the fidelity of the recovered state.

The rest of the paper is organized as follows. In Sec.~\ref{sec:prelim}, we introduce the code-specific Petz map and briefly discuss its role as a near-optimal, noise-adapted recovery map. In Sec.~\ref{sec:petz_circuit}, we describe our first approach to implementing the Petz map using its isometric extension. In Sec.~\ref{sec:petz_povm}, we describe a POVM-based approach to implement the Petz map. Finally, in Sec.~\ref{sec:petz_qsvt}, we describe our circuit construction that combines the block encoding technique with isometric extension. We discuss the example of the $4$-qubit code subject to amplitude-damping noise in Sec.~\ref{sec:petz results} and obtain fidelities corresponding to the different implementations of the Petz map in this case. We conclude with a summary and future outlook in Sec.~\ref{sec:summary}.

\section{Preliminaries}\label{sec:prelim}

We begin with a brief review of the Petz recovery map. Recall that a quantum noise channel is modelled as a completely positive trace-preserving (CPTP) map $\cE$ acting on the set of states $\cS(\cH)$ of a Hilbert space $\cH$. Such a map is characterized by a set of Kraus operators $\{E_{i}\}$, satisfying $\sum_{i=1}^{N}E_{i}^{\dagger}E_{i} = I$. 

\subsection{The Petz map}\label{sec:petz}

The Petz map was first defined in the context of reversing noisy dynamics on a specific system state~\cite{barnum2002}. The action of the \emph{state-specific} Petz map corresponding to a state $\sigma \in \cS(\cH)$ and noise map $\cE$, denoted as $\cR_{\sigma, \cE}$, is given by,
\begin{align}
    \cR_{\sigma, \cE} (.) &= \sqrt{\sigma} \, \cE^{\dagger} \,\left(\cE(\sigma)^{-1/2}(.)\cE(\sigma)^{-1/2}\right) \, \sqrt{\sigma}, \nonumber \\
    &=  \sqrt{\sigma}  \sum_{i=1}^{N} E_{i}^{\dagger}\left(\cE(\sigma)^{-1/2}(.)\cE(\sigma)^{-1/2} \right) E_{i} \sqrt{\sigma} . \label{eq:petz_state}  
\end{align}
Here, $\cE(\sigma) = \sum_{i}E_{i}\sigma E_{i}^{\dagger}$ denotes the action of the channel on the state $\sigma$ and  $\cE^{\dagger}\sim \{E_{i}^{\dagger}\}$ is the adjoint map corresponding to the map $\cE$. The map $\cR_{\sigma,\cE}$ is clearly completely positive (CP) and trace non-increasing; it can be normalised to make it trace-preserving (TP) as well.

From the definition in Eq.~\eqref{eq:petz_state}, it is clear that the state-specific Petz map \emph{perfectly} reverses the action of the noise $\cE$ on the state $\rho$. It was shown in~\cite{barnum2002} that for an ensemble of states $\{p_{i},\rho_{i}\}$, affected by noise $\cE$, the Petz map $\cR_{\rho, \cE}$ corresponding to the state $\rho = \sum_{i}p_{i}\rho_{i}$ is a near-optimal recovery map. Here, optimality was defined in terms of the average entanglement fidelity between the ideal and noisy states. 

In our work, we focus on the \emph{code-specific} Petz map defined in~\cite{hkn_pm2010}. Recall that an $[n,k]$ quantum code $\cC$ encoding $k$ qubits in $n$, is a $2^{k}$-dimensional subspace of the $n$-qubit space.  Corresponding to a quantum code $\cC$ and noise map $\cE$, the Petz recovery channel is defined as, 
\begin{align}
\cR_{P, \cE} (.) &= P\,\cE^{\dagger} \,(\cE(P)^{-1/2}(.)\cE(P)^{-1/2})\,P \label{eq:petz} \\
&= P \sum_{i=1}^{N} E_{i}^{\dagger}\left(\cE(P)^{-1/2}(.)\cE(P)^{-1/2} \right) E_{i} P, \nonumber
\end{align}
where $P$ is the projector onto the codespace $\cC$ and $\cE(P) = \sum_{i}E_{i}PE_{i}^{\dagger}$ denotes the action of the noise map on the codespace. The map $\cR_{P,\cE}$ can be thought of as a special case of the map $\cR_{\sigma,\cE}$, where the state $\sigma$ is chosen to be the maximally mixed state $(P/2^{k})$ on the codespace. Thus, $\cR_{P,\cE}$ is also clearly completely positive (CP) and can be normalised to make it trace-preserving (TP) as well. Furthermore, the composite map $\cR_{P,\cE}\circ\cE$, which denotes the action of noise followed by Petz recovery, is unital and preserves the identity operator on the codespace. We also note that if the noise channel $\cE$ is perfectly correctable on a codespace $\cC$, the corresponding code-specific Petz recovery is indeed the same as the standard (perfect) recovery map~\cite{hkn_pm2010}.

It was shown in~\cite{hkn_pm2010} that $\cR_{P,\cE}$ is a near-optimal recovery map for the codespace $\cC$ under the action of noise $\cE$, where optimality is defined in terms of the \emph{worst-case fidelity} on the codespace. Recall that the fidelity measure $F$, based on the Bures metric, between a pure state $|\psi\rangle$ and a mixed state $\rho$ is defined as~\cite{nielsen},
\begin{equation}
    F^{2}(|\psi\rangle, \rho) = \langle \psi | \rho |\psi\rangle . \label{eq:fidelity}
\end{equation}
The worst-case fidelity corresponding to codespace $\cC$ and the composite map $\cR_{P,\cE}\circ\cE$ is then obtained by minimizing the fidelity over all states in the codespace, namely,
\begin{equation}
    F^{2}_{\rm min} (\, \cC, \cR_{P,\cE}\circ\cE \, ) = \min_{|\psi\rangle \in \cC} \langle \psi | \cR_{P,\cE}\circ\cE (|\psi\rangle\langle \psi|)| \psi\rangle. \label{eq:wc_fidelity}
\end{equation}
The worst-case fidelity measure thus guarantees that the \emph{all} states in the codespace are preserved to a high degree. The near-optimality result in~\cite{hkn_pm2010} states that the Petz map defined in Eq.~\eqref{eq:petz} achieves a worst-case fidelity that is close to that of the optimal recovery map corresponding to code $\cC$ and noise $\cE$.

\subsection{Quantum algorithm for the Petz map}\label{sec:petz_algo}

It was recently shown in~\cite{gilyen2022_petz} that a quantum circuit that \emph{approximately} implements the state-specific Petz map in Eq.~\eqref{eq:petz_state} can be realised using the techniques of block encoding~\cite{low2019} and the quantum singular value transform (QSVT)~\cite{Gily_n_2019}, by decomposing the Petz map as sequence of completely positive (CP) maps. Note that the maps in Eq.~\eqref{eq:petz_state} and Eq.~\eqref{eq:petz} can be thought of as being composed of three (CP) maps. For example, the state-specific map $\mathcal{R}_{\sigma,\cE}$ can be decomposed in terms of (a) the adjoint map $\cE^{\dagger} (.)$ with Kraus operators $\{E_{i}^{\dagger}\}$, (b) the normalization or amplification map $\cE(\sigma)^{-1/2}(.)\cE(\sigma)^{-1/2}$ and (c) the projection map $\cP_{\sigma}(.) = \sqrt{\sigma}\,(.) \sqrt{\sigma}$. These maps are trace-increasing in general, but the overall map obtained by composing these three maps is indeed trace-preserving.

The quantum algorithm to implement the state-specific Petz recovery map $\cR_{\sigma, \cE}$ described in~\cite{gilyen2022_petz}, proceeds via the following steps. First, it assumes that there exist quantum circuits that realise approximate \emph{block encodings} of the state $\sigma$ and the operator $\cE(\sigma)$. Furthermore, it also assumes the existence of a quantum circuit that implements an \emph{isometric extension} of the adjoint map $\cE^{\dagger}(.)$. Finally, the algorithm uses the QSVT and the technique of amplitude amplification~\cite{berry2014_amp} to realise the $\cE(\sigma)^{-1/2}$ and the $\sqrt{\sigma}$ operations. In what follows, we will briefly review the concepts of isometric extension and block encoding, since they are used in some of our circuit constructions as well.

\subsubsection{Isometric Extension of a CPTP map}

Consider a CPTP map $\cE: \cS(\cH_{A})\rightarrow \cS(\cH_{B})$ mapping states on Hilbert space $\cH_{A}$ to states on $\cH_{B}$. Let $\cH_{E}$ denote the Hilbert space of an ancillary system $E$ and $\rho_{B} \in \cS(\cH_{E})$ denote some fixed state of the ancilla. The isometric extension of $\cE$ is defined to be the unitary $U^{\cE}_{AE \rightarrow B}$ mapping states on $\cH_{A} \otimes \cH_{E}$ to states on $\cH_{B}$, such that, 
\begin{equation}
    \tr_E [U^{\cE}_{AE\rightarrow B}(\rho_A \otimes \rho_{E})(U^{\cE}_{AE\rightarrow B})^{\dagger}] = \cE(\rho_A). \label{eq:iso_extn}
\end{equation}
The isometric extension of a map $\cE$ can be implemented via a quantum circuit as shown in Fig.~\ref{fig:isometric_map}. 

\begin{figure}
    \centering

\includegraphics[width = 0.55 \columnwidth]{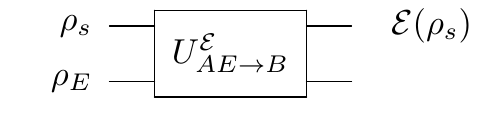}
    \caption{Isometric extension circuit of a map $\cE$}
    \label{fig:isometric_map}
\end{figure}

However, the adjoint map $\cE^{\dag}$ is not a physical map in general and hence does not admit an isometric extension. However, it was shown in~\cite{gilyen2022_petz} that the adjoint map can be implemented via a quantum circuit, starting with the isometric extension for the forward channel $\cE$ and using the procedure of block encoding described below.

\subsubsection{Block Encoding}

The idea of block encoding is to represent any arbitrary matrix $A$ as the top-left block of a unitary matrix. If the unitary $U^A$ is an \emph{exact} block encoding of $A$, then it has the following structure.
\begin{align}\label{eq:7}
     U^A &= \begin{bmatrix}A & \quad . \\
                            . & \quad. \end{bmatrix} .
 \end{align}
We formally define the block encoding for an $s$-qubit operator here and note that this can be generalized to arbitrary dimensions as well~\cite{Gily_n_2019}. An $(s+a)$-qubit unitary $U$ is said to be an $(\alpha,a,\epsilon)$-block encoding of an $s$-qubit operator $A$, if,
 $$ \parallel A- \alpha(\langle0|^{\otimes a} \otimes I_{2^s} ) U ( |0\rangle^{\otimes a} \otimes I_{2^s} ) \parallel \; \leq \; \epsilon, $$
where $\parallel (.) \parallel$ denotes the operator norm, $\alpha,\epsilon \in \mathbb{R}_+$ and $a\in \mathbb{N}$. $U$ is said to be an exact block encoding of the matrix $(A/\alpha)$ if $\epsilon=0$. The unitary $U$ is thus an exact block encoding of $A$ if $\alpha=1$ and $\epsilon=0$. 
In order to get the action of the $s$-qubit operator $A$, we measure the $a$-qubit ancilla and keep the $s$-qubit output state only if the $a$-qubit ancilla returns the all-zero after the measurement. This procedure is depicted in Fig \ref{fig:blenc_Agen}.

Of particular interest here is the existence of unitaries that provide block encodings of density operators. In this context, we note the following result from \cite{Gily_n_2019}, which shows that an exact block encoding for any density operator $\rho$ can be constructed via a \emph{purifying} isometry on an extended space.
\begin{lemma}\label{lem:a}
Suppose that $ \rho$ is an $s$-qubit density operator and $G$ is an (a + s)-qubit unitary that when acting on the $|0\rangle|0\rangle$ input state gives a purification $|0\rangle|0\rangle \rightarrow |\rho\rangle$ on the $a+s$-qubit system, such that, \, $\tr_a |\rho\rangle |\rho\rangle=\rho$. Then the operator $U$ given by
    \begin{equation}
        U = (G^{\dagger}\otimes I_{2^s}) (I_{2^a} \otimes SWAP_s) (G \otimes I_{2^s}),
    \end{equation}
is an exact block encoding of the density matrix $\rho$. 
\end{lemma}

\begin{figure}[t]
    \centering
 
  \includegraphics[width = 0.55 \columnwidth]{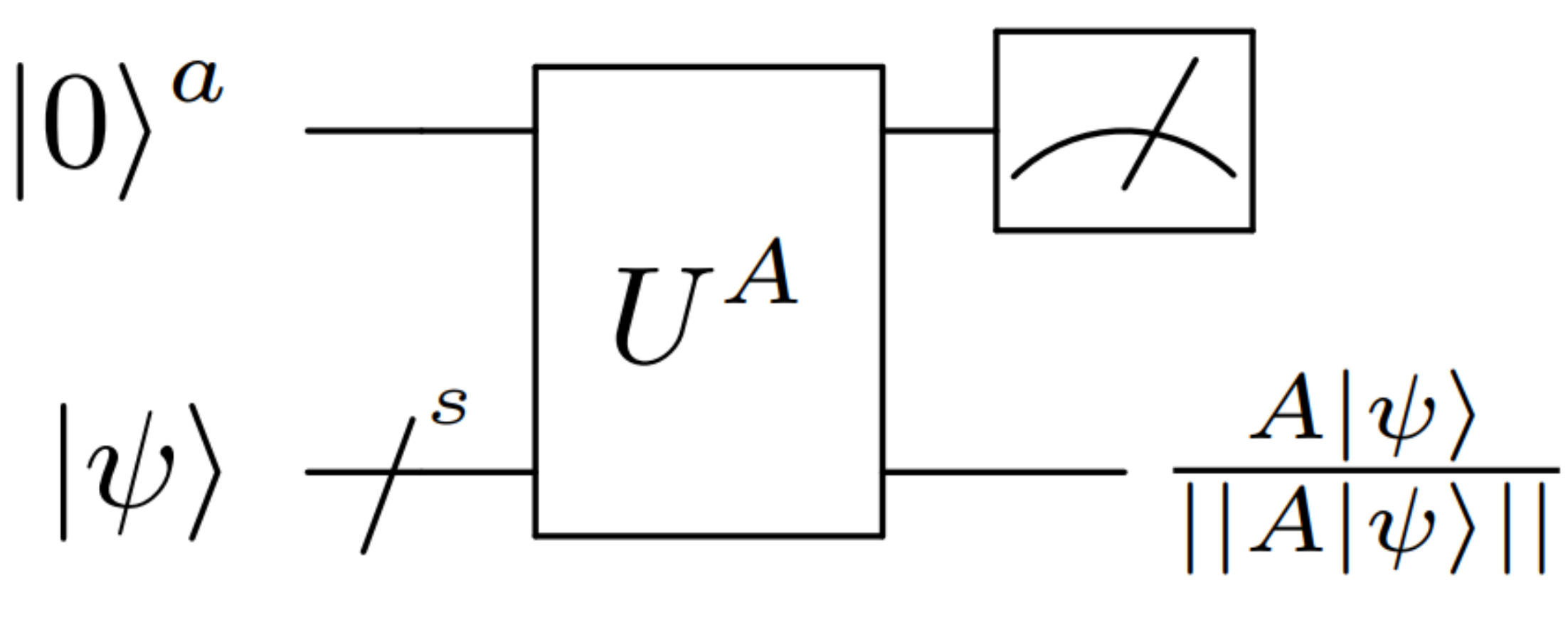}
    \caption{The $(a+s)$- qubit block encoding of A. $U^A$ returns $A|\psi\rangle$ upon measuring the $a$-qubit ancilla, whenever the outcome is the all-zero state. }
    \label{fig:blenc_Agen}
\end{figure}

\section{Petz recovery circuit via isometric extension}\label{sec:petz_circuit}

Arguably, the most direct approach to implementing the code-specific Petz recovery map $\cR_{P,\cE}$ is via its isometric extension. However, this is rather resource intensive in general, requiring, for example, of the order of $dN$ two-level unitaries~\footnote{A two-level unitary is one that acts nontrivially only on a two-dimensional subspace of the $d$-dimensional space. We refer to~\cite{nielsen} for further details.}, for a quantum channel with $N$ Kraus operators acting on a $d$-dimensional space. In what follows, we first present a resource-efficient implementation of the isometric extension of any CPTP map. We then construct a circuit that can directly estimate the fidelity between an arbitrary initial state and the final state after the action of a quantum channel. Both these results are then applied to the specific case of the Petz map $\cR_{P,\cE}$.

\subsection{Isometric extension of the Petz map using two-level unitaries}\label{sec:iso petz}
 
Consider a quantum channel $\cR$ with $N$ Kraus operators acting on a $d$-dimensional Hilbert space $\cH_{S}$. We consider a specific isometric extension of $\cR$, denoted as $\cV_{\cR}$, where the ancillary system starts in a fixed pure state $|0\rangle$, so that,

 \begin{align}\label{eq:iso_fixed}
          \cR(\rho) &= \tr_E [\cV_{\cR}( \rho_{S}\otimes |0_E\rangle \langle0_E|)V^{\dag}_{\cR}] .
      \end{align}
If $\{R_i, \; i \in [0,\ldots, N-1] \}$ denote the Kraus operators of the quantum channel $\cR$, then the isometry $\cV_{\cR}$ has the following block matrix structure on the extended space $\cH_{S} \otimes \cH_{E}$.
        \begin{align}\label{eq:isometry1}
          \cV_{\cR} &= \begin{bmatrix} R_0 & \hdots & \hdots\\
                                       \vdots & \ddots & \vdots\\
                                       R_{N-1} & \hdots& \hdots.
                                     \end{bmatrix}
        \end{align}

Note that the unitary $\cV_{\cR}$ acts on a space of dimension $dN$. The first $N$ blocks of the unitary $\cV_R$ contain the Kraus operators, and the rest of the unitary does not affect the system of interest.

To construct a quantum circuit to implement $\cV_{\cR}$ in Eq.\eqref{eq:isometry1}, we will follow the well-known procedure of decomposing the unitary into a product of two-level unitaries~\cite{nielsen}. According to this prescription, the unitary operator $\cV_{\cR}$ can be realised as a product of $k$ two-level unitaries $\{V_{1}, \ldots, V_{k}\}$ as follows.
\begin{align}
    \cV_{\cR} &= V_1 V_2 \ldots V_k \nonumber \\
    ( V_1 V_2 \ldots V_k)^{\dag} \cV_{\cR} &= I_{dN} \nonumber \\
    \mbox{ where }\qquad k &= \frac{dN (dN-1)}{2}. \label{eq:two-level1}
\end{align}
Here, $I_{dN\times dN}$ denotes the identity operator on a $dN$-dimensional space. 

Consider now, the Petz recovery map $\cR_{P, \cE}$ corresponding to a noise channel $\cE$ with $N$ Kraus operators and an $n$-qu\emph{d}it code $\cC$ with projector $P$. Note that $\cE$ here refers to the single-qu$d$it noise channel with $N$ Kraus operators, which acts in an i.i.d. fashion on the $n$ qu$d$its of the codespace. The Petz recovery map, in this case, contains $N^{n}$ Kraus operators, each acting on a $d^{n}$-dimensional space. As per the upper bound in Eq.~\eqref{eq:two-level1}, the isometric extension circuit for this recovery map could require as many as $d^{n}N^{n}(d^{n}N^{n}-1)/2$ two-level unitaries.  

As we argue below, for our purposes, we do not need to decompose the complete unitary $\cV_{\cR}$. Rather, it suffices to decompose $\cV_{\cR}$ partially in such a way that the isometric extension circuit can be realised with only $d^{2n}N^n$ two-level unitaries. 

\begin{lemma}\label{lem:Petz_isometry}
The Petz map $\cR_{P,\cE}$ corresponding to a noise channel $\cE$ with $N$ Kraus operators and an $n$-qu\emph{d}it code $\mathcal{C}$ with projector $P$, can be implemented exactly using a quantum circuit comprising $d^{2n}N^{n}$ two-level unitaries and $n$ ancillary qu$d$its.
\end{lemma}
\begin{proof}
According to the two-level unitary decomposition procedure in \cite{nielsen}, we note that for any unitary matrix $U_{D\times D}$ on a $D$-dimensional system we can find $m$ two-level unitaries such that, 
\begin{align}\label{eq:partial_two_level}
    \prod\limits_{i=1}^{m} V_i^{\dag} \, U_{D\times D} &= \begin{pmatrix}
              I_{m} & 0 \\
                   0                 & M_{D-m \times D-m}
    \end{pmatrix} .
\end{align}
For a given $n$-qu\emph{d}it code and a noise channel $\cE$ with $N$ Kraus operators, we know all the $N$ Kraus operators of the Petz recovery $\cR_{P,\cE}$, as defined in Eq.~\eqref{eq:petz}. Let $\{R_{i}\}$ denote the Kraus operators of $\cR_{P,\cE}$. Since these are operators of dimension $d^n \times d^n$, we essentially know the first $d^n$ columns of the isometric extension $\cV_{\cR_{P,\cE}}$. Therefore, according to the Eq.\eqref{eq:partial_two_level}, there exist $N^nd^{2n}$ two-level unitaries $\{ V_1, V_2, \ldots , V_{d^n}\}$ such that,
\begin{align}
    ( V_1 V_2 \ldots V_{N^nd^{2n}})^{\dag} V_{\cR_{P,\cE}} &= \begin{pmatrix} \mathbb{I}_{d^{n}\times d^n} & 0 \\ 0 & * \end{pmatrix}. \label{eq:iso2}
\end{align}

Let the product $ ( V_1 V_2 \ldots V_m)$ be denoted $\Tilde{V}_{\cR_{P,\cE}}$. Then Eq.\eqref{eq:iso2} implies, 
\begin{align}
    \left(\Tilde{V}_{\cR_{P,\cE}}\right)^{\dag} &= \begin{pmatrix} R_0^{\dag} & R_1^{\dag} & \hdots &R_{N-1}^{\dag}\\
                                             \vdots & \vdots & \ddots &\vdots\end{pmatrix}. 
\end{align}
Since the Petz map is a CPTP map, its Kraus operators $\lbrace R_i \rbrace$ satisfy $ \sum\limits_{i=0}^{N-1}\, R_i^{\dag}R_i = \mathbb{I}_{d^n \times d^n}$, thus showing that the action of  unitary $\Tilde{V}_{\cR_{P,\cE}}$ is equivalent to the action of the unitary $\cV_{\cR_{P,\cE}}$ on the first $d^{n}\times d^{n}$ block which is the domain of the Petz map in this case. 
\begin{align}
    (\Tilde{V}_{\cR_{P,\cE}})^{\dag}   \cV_{\cR_{P,\cE}} &= \begin{pmatrix} R_0^{\dag}  & \hdots &R_{N-1}^{\dag}\\
                                             \vdots  & \ddots &\vdots\end{pmatrix} \begin{bmatrix} R_0 & \hdots \\
                                       \vdots & \ddots \\
                                       R_{N-1} & \hdots
                                     \end{bmatrix} \nonumber \\
                                     &= \begin{pmatrix} \sum_{i=0}^{N-1} R_i^{\dag}R_i & 0 \\ 0 & * \end{pmatrix} \nonumber \\
                                     &= \begin{pmatrix} I_{d^n} & 0 \\ 0 & * \end{pmatrix}. \label{eq:iso3}
\end{align}
Putting together Eqs.~\eqref{eq:iso2} and Eq.~\eqref{eq:iso3}, it follows that the Petz map can be exactly realised using $d^{2n}N^n$ two-level unitaries. 

\end{proof}

Note that the reduction in the number of two-level unitaries shown in Lemma~\ref{lem:Petz_isometry} holds more generally for a quantum channel whose isometric extension is defined using a fixed pure state of the ancilla, as in Eq.~\eqref{eq:iso_fixed}. The final step in implementing the isometric extension $\cV_{\cR_{P,\cE}}$ is to identify the set of two-level unitaries $\{V_{1}, V_{2},\ldots, V_{N^nd^{2n}}\}$ that can be done following, for example, the standard procedure described in~\cite{nielsen}. 

Fig.~\ref{fig:std_comb} shows a complete circuit that combines the isometric extension of the noise channel $\cE$ and that of the recovery channel $\cR_{P,\cE}$, along with the encoding unitary $U_{\rm en}$. Note that the Petz recovery circuit shown here for an $n$-qu$d$it code requires $n$ ancillary qubits initialized to the $|0\rangle$ state. 

\begin{figure}[t]
    \centering

\includegraphics[width = 0.9 \columnwidth]{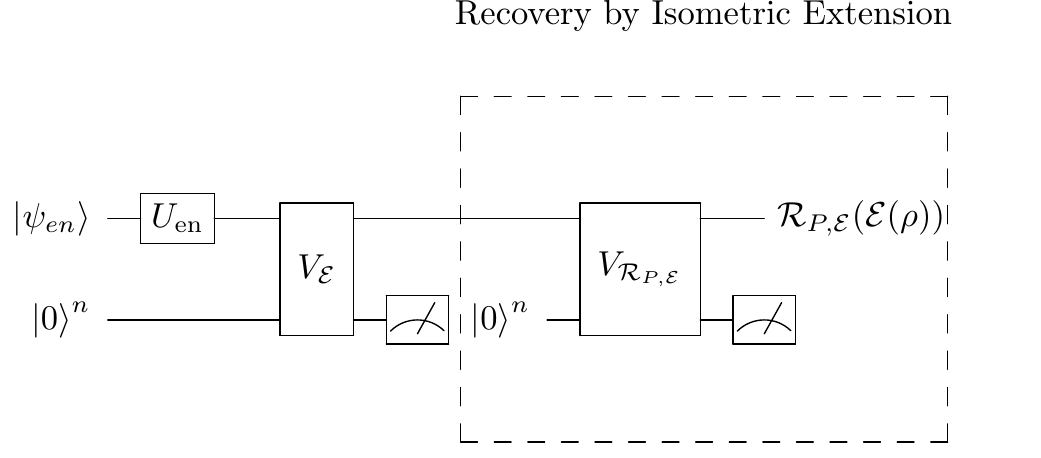}
    \caption{Circuit for Petz recovery $\cR_{P,\cE}$, after the action of noise $\cE$ (implemented via the unitary $V_{\cE}$) and encoding $U_{\rm en}$, using the isometric extension $V_{\cR_{P,\cE}}$.}
   \label{fig:std_comb}
\end{figure}

\subsection{Fidelity Calculation Via Block encoding }\label{sec:fid_blockEnc}
We next demonstrate how the technique of block encoding can be used to construct a quantum circuit that directly estimates the fidelity between a given (initial) state and the final state after the action of noise and recovery. This provides a way to circumvent the need for full-state tomography, which is often expensive to implement, and instead estimate the fidelity in a more direct, resource-efficient way. 

Our approach is distinct from other works in the literature which have proposed quantum algorithms for computing the fidelity between a pair of arbitrary density operators~\cite{cerezo2020, ag_2020_fid,wang_fid}. These algorithms however, assume knowledge of both the density operators and are estimate a function that approximates the actual fidelity function. Here, we address the question of estimating the fidelity between a given $n$-qubit (initial) state $|\psi\rangle$ and the density operator that results after the action of a CPTP map $\cN$ on the state $|\psi\rangle$. Our circuit construction makes use of the isometric extension of the map $\cN$ and a unitary matrix $U$ that prepares the initial state $|\psi\rangle$ from $|0\rangle^{\otimes n}$.\\
Recall that the fidelity between an initial state $|\psi\rangle$ and the corresponding final state $\cN(|\psi\rangle \langle\psi|)$ is calculated as defined in Eq.~\eqref{eq:fidelity}. Instead of calculating the fidelity by performing full state tomography on the output state $\rho = \cN(|\psi\rangle\langle\psi|)$ and then using Eq.~\eqref{eq:fidelity}, we propose a direct method for estimating $F^2$. Specifically, we construct a circuit whose output gives us the fidelity $F^2$ as the probability of getting the all-zero state when measured in the computational basis. 
\begin{theorem}
 Consider the action of a CPTP map $\cN$ on an $n$-qubit pure state $|\psi\rangle$. Let $V_{\cN}$ denote the isometric extension $\cN$ and $U$ be a unitary such that $U|0\rangle^{\otimes n} = |\psi\rangle$. Then, the circuit in Fig.\ref{fig:purf_blenc_1} can be used to estimate the fidelity $F^{2} = \langle \psi|\cN(|\psi\rangle\langle\psi|)|\psi\rangle$ as follows. Upon measuring all the registers in Fig.~\ref{fig:purf_blenc_1} in the computational basis, the probability of getting the all-zero state as the outcome is given by $F^{4}$.
 \end{theorem}\label{thm:fidelity}
\begin{proof}
To calculate the fidelity, we first encode the $n$-qubit density matrix $\rho = \cN (|\psi\rangle\langle\psi|)$) as the top-left block of a unitary matrix, following the procedure in Lemma~\ref{lem:a}. Recall that exact block encoding of an $n$-qubit density matrix $\cN(|\psi\rangle\langle{\psi}|)$, requires an $(n+s)$-qubit \emph{purification unitary}  $G$ for the state $\cN(|{\psi}\rangle\langle{\psi}|)$, such that,
 \begin{equation}
     G |{0}\rangle^{\otimes n}|{0}\rangle^{\otimes s} = |\cN(|{\psi}\rangle\langle{\psi}|)\rangle, \label{eq:purf_1}
 \end{equation} 
 where $ |\cN(|{\psi}\rangle\langle{\psi}|)\rangle$ denotes a purification of the density operator $\cN(|{\psi}\rangle\langle{\psi}|)$. This purification circuit is schematically depicted in Fig.~\ref{fig:purf_1}.

\begin{figure}[t]
    \centering
  \includegraphics[width = 0.6 \columnwidth]{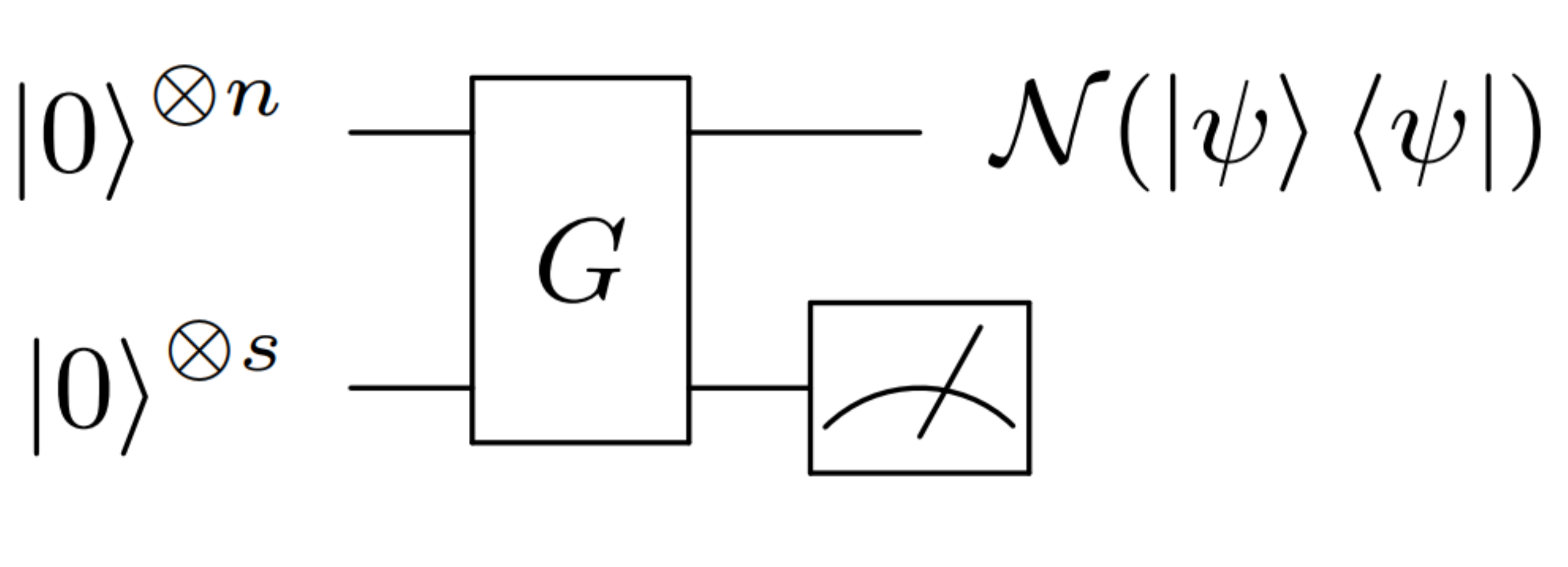}
     \caption{Purification circuit for $\cN(|\psi\rangle\langle\psi|)$.}
     \label{fig:purf_1}
 \end{figure}
 Note that the purification circuit in Fig.~\ref{fig:purf_1} is similar to the isometric extension circuit described in Fig.~\ref{fig:isometric_map}, with the only difference being in the specific choice of input states. We can make these two circuits effectively identical just by adding another unitary $U$ which acts on $|0\rangle^n$ to give the state $|\psi\rangle$, along with identifying that $s=n$. Therefore the purification $G$ is nothing but a matrix product of the isometric extension  $V_{\cN}$ of the channel $\cN$ and $I_{2^{n}\times 2^{n}}\otimes U$ as depicted by the boxed circuit in Fig.~\ref{fig:purf_blenc_1}. 
 
\begin{figure}[t]
    \centering
\includegraphics[width = 0.9 \columnwidth]{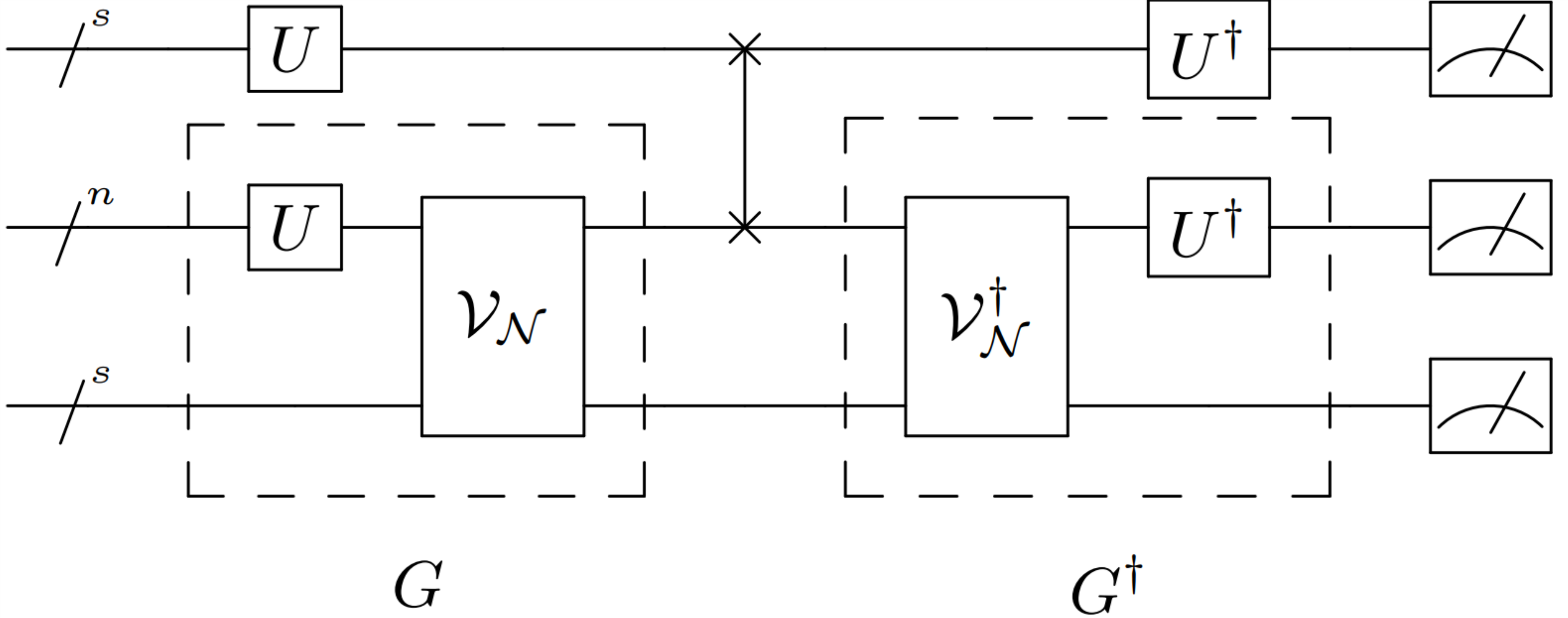}
    \caption{Circuit for the fidelity calculation. The boxed part is the purification unitary $G$ for the state $\cN(|{\psi}\rangle\langle{\psi}|)$.}
    \label{fig:purf_blenc_1}
\end{figure}

Now we analyse the rest of the circuit in Fig.~\ref{fig:purf_blenc_1}.
As explained above, the output of the boxed circuit is the block encoding,
      \begin{equation}
          U^{\cN(|{\psi}\rangle\langle{\psi}|)} = \begin{pmatrix}\cN(|{\psi}\rangle\langle{\psi}|) & * \\ * & * \end{pmatrix}.
      \end{equation}
Therefore, the unitary $(\mathbb{I} \otimes U^{\dag})   U^{\cN(|{\psi}\rangle\langle{\psi}|)} (\mathbb{I} \otimes U)$ 
has the form, 
      \begin{align}\label{eq:proof_1}
        U_{\rm Fid} &= \begin{pmatrix} U^{\dag}\cN(|{\psi}\rangle\langle{\psi}|)U & * \\ * & *\end{pmatrix}.
      \end{align}
    
The action of the unitary $U_{\rm Fid}$ on the input state is, 
   
     \begin{equation*} 
        U_{\rm Fid}|{0}\rangle^{ 3n}= U^{\dag}\cN(|{\psi}\rangle\langle{\psi}|) (|{\psi}\rangle \otimes|{0}\rangle^{ 2n}) +|{*}\rangle.
        \end{equation*}
        Thus the first matrix element of $U_{\rm Fid}$ in the computational basis evaluates to,
        \begin{equation}
         \langle {0}^{3n} | U_{\rm Fid} | 0^{3n} \rangle   = \langle{\psi}|\cN(|{\psi}\rangle\langle{\psi}|)|{\psi}\rangle .
     \end{equation}
     Therefore the probability of getting all-zero state after performing a measurement in the computational basis is 
     \begin{align*}
         \mbox{Prob}(|{0}\rangle^{3n}) &= \langle{\psi}|\cN(|{\psi}\rangle\langle{\psi}|)\|{\psi}\rangle^2 = F^{4}.
     \end{align*}

\end{proof}

We now use the circuit construction outlined in Theorem~\ref{thm:fidelity} to obtain a circuit that evaluates the fidelity between any encoded state and the state after the action of noise followed by Petz recovery. Consider a noise channel $\cE$ acting on an $n$-qubit code, followed by the Petz recovery map $\cR_{P,\cE}$. Let $\{E_{i}\}$ denote the $N$ Kraus operators of the $n$-qubit noise channel. Correspondingly, let $\{R_i\}$ denote the $N$ Kraus operators of the $n$-qubit Petz recovery map. Note that both sets of Kraus operators $\{E_i\}$ and $\{R_i\}$ are of dimension $2^n \times 2^n$.

Given the Kraus operators, the individual isometric extensions  $V_{\cE}$ and  $V_{\cR_{P,\cE}}$ are known. We can then construct the isometric extension for the composite channel $(\cR_{P,\cE})\circ\cE$ by combining these individual isometric extensions as shown in Fig.~\ref{fig:G2}. Here, the unitary $V_{\cE}$ is applied only when the control block is set to $|0\rangle^{n}$.

 \begin{figure}[t]
        \centering
 \includegraphics[width = 0.7 \columnwidth]{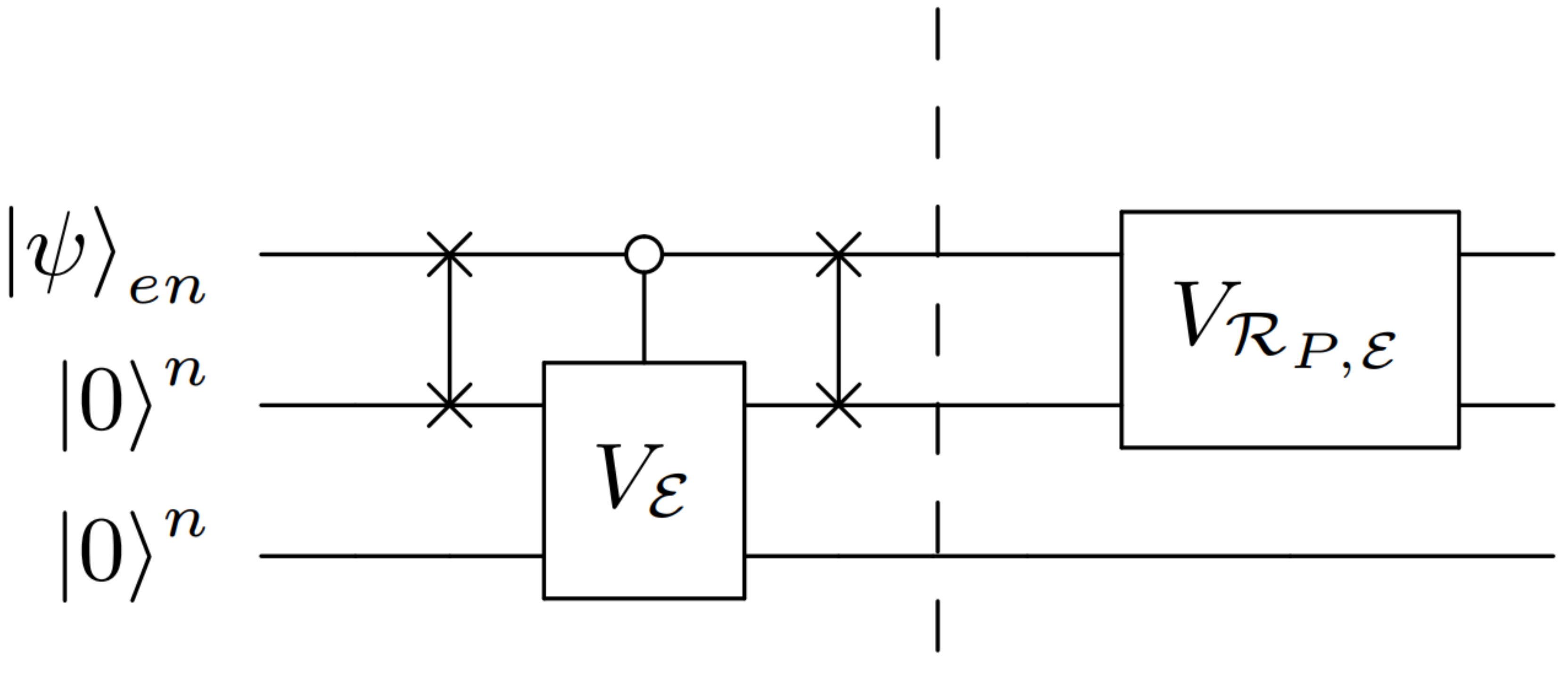}
        \caption{Isometric extension for the composite channel $(\cR_{P,\cE})\circ \cE$ using the individual isometric extensions.} 
        \label{fig:G2}
    \end{figure}

Now, let $|{\psi_{en}}\rangle$ be a state in the codespace of an $n$-qubit code, which is prepared by the action of the encoding unitary $U_{en}$ on $|{\psi}\rangle\otimes|{0}\rangle^{n-1}$. Since our circuit requires the input to be the all-zero states, we construct another unitary $\tilde{U}_{\rm en}$ such that, 
\[ \tilde{U}_{en} |{0}\rangle^{n} = |{\psi_{en}}\rangle.\]  
We demonstrate an example of such a construction of $\tilde{U}_{en}$ for a specific QEC code in Sec.\ref{sec:petz results}. Then, it follows from Thm.~\ref{thm:fidelity} that the circuit in the Fig.~\ref{fig:bl_enc_rec_1} evaluates the fidelity between the density matrix $(\cR\circ\cE)(|\psi_{\rm en}\rangle \langle\psi_{\rm en}|)$ and the initial encoded state $|\psi_{en}\rangle$, with all the inputs initialized to the zero state. The boxed part of the circuit in Fig.~\ref{fig:bl_enc_rec_1} is the purification of the density matrix $(\cR\circ\cE)(|\psi_{\rm en}\rangle \langle\psi_{\rm en}|)$.

\section{Petz map circuit via Polar Decomposition}\label{sec:petz_povm}

Our second approach to construct a circuit for the Petz map invokes general quantum measurements corresponding to positive operator valued measures (POVMs). Our approach is based on the fact that the Kraus operators of any channel can be \emph{approximately} realised via a sequence of binary outcome POVMs and unitary operations, as shown in Lemma~\ref{lem:petz_povm} below. 

The fact that quantum channels can be realised using unitary operators and POVMs is well known~\cite{lloyd2001engineering, nielsen}. The original proposal in~\cite{lloyd2001engineering} uses a $N$-outcome quantum measurement to implement a CPTP map with $N$ Kraus operators, but this implementation typically requires $\lceil \log_{2}N \rceil$ ancillary qubits~\cite{yordanov2019implementation}. Alternate approaches to implementing $N$-outcome POVMs that require only a single ancillary qubit are known, but these approaches either involve non-Hermitian quantum measurements~\cite{lloyd2001engineering} or require extending the system Hilbert space via Naimark dilation~\cite{andersson2008}. 

Here, we adapt some of these existing approaches to approximate a quantum channel via a sequence of two-outcome POVMs and unitary operations, requiring only one additional ancillary system. A similar approximate approach to simulating open quantum system dynamics has been outlined in~\cite{gupta2020optimal}, but this approach relies on Stinespring dilation rather than POVMs. Finally, we note that a related approach to implementing quantum channels using a sequence of POVMs and a single ancillary qubit has also been proposed in~\cite{shen2017}.

\begin{figure}[t]
          \centering
   \includegraphics[width = 1.0 \columnwidth]{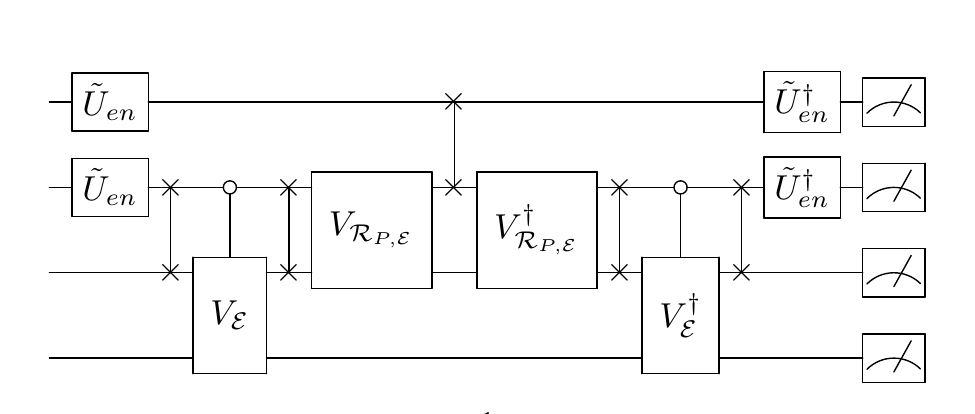}
          \caption{Circuit for the fidelity calculation}
        \label{fig:bl_enc_rec_1}
      \end{figure} 
      
\begin{lemma}\label{lem:petz_povm}
A CPTP map $\cK$ of rank $N$, described by $N$ Kraus operators $\cK \sim \sum_{i =1}^{N}\{K_{i}\}$, can be decomposed using a sequence of quantum operations $\{\cR_{j}\}$, as, 
\begin{equation}
    \label{eq:decomp}
    \cK(\rho)  = ( \cR_{N} \circ \cR_{N-1} \cdots \cR_{1} ) (\rho) + \mathcal{O}(\parallel K_{i}^{\dagger}K_{i}\parallel^{2}),
\end{equation}
where each operation $\cR_{j}$ is a composition of a $2$-outcome POVM and  a unitary gate that is controlled by the outcome of the POVM.
\end{lemma}

\begin{proof}
We start with the polar decomposition of the Kraus operators. Let $\{U_{i}\}$ be a set of unitary matrices and $\{P_{i}\}$ be a set of positive semi-definite operators, such that,
\begin{equation}
\label{eq:polar}
    K_{i} = U_{i}P_{i} = U_{i}\sqrt{K_{i}^{\dagger}K_{i}}, \; \forall \; i\in[1,N].
\end{equation}

Since $\cK$ is a trace-preserving map, the set $\{P_{i}\}$ defines an $N$-outcome POVM.
\begin{equation}
    \label{eq:TP}
    \sum_{i = 1}^{N}P_{i}^2 = \sum_{i= 1}^{N}K_{i}^{\dagger}K_{i} = I.
\end{equation}

\begin{widetext}
\begin{center}
\begin{figure}[t]
\centering
\includegraphics[scale=0.29]{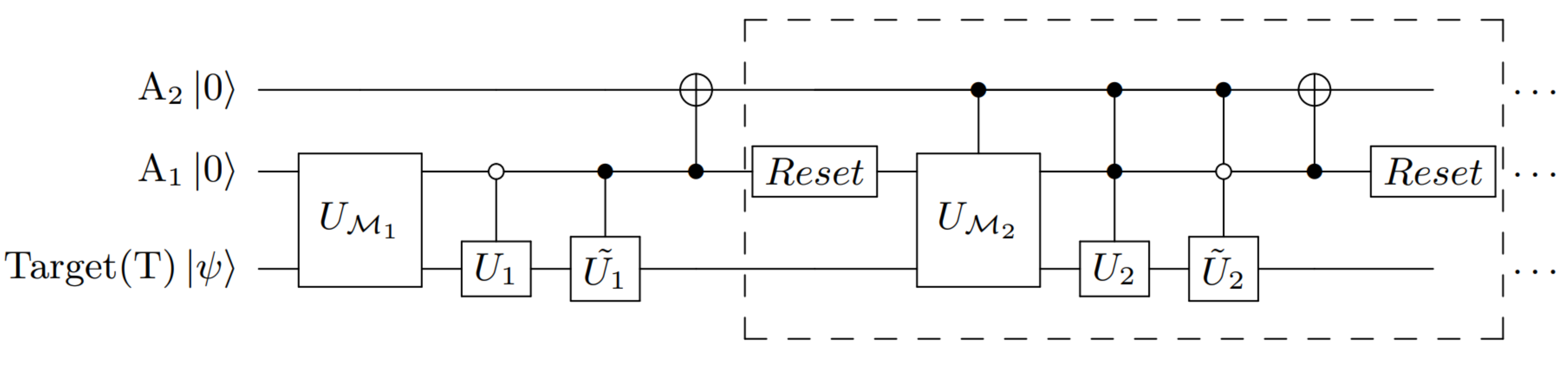}
    \caption{The first two steps of the quantum circuit used to implement any CPTP map as described in Lemma~\ref{lem:petz_povm}. The unitary gates labelled $U_{\cM_{i}}$ implement the POVMs $\cM_{i}$ at each step. The two conditional unitaries following $U_{\cM_{i}}$ implement the operations $U_i$ and $\tilde{U_i}$ as described in Eq.~\eqref{eq:polar}. The second component of the circuit (boxed) is repeated $N-2$ times in order to implement a rank $N$ CPTP map.}
    \label{fig:blockdiagram}
\end{figure}
\end{center}
\end{widetext}

We now define the sequence of operations $\{\cR_{i}\}$ as follows. We first construct $N$ two-outcome POVM measurements $\cM_{i}$, each defined by the pair of positive operators $P_{i} =\sqrt{K_{i}^{\dagger}K_{i}}$ and $Q_{i} = \sqrt{I - K^{\dagger}_{i}K_{i}}$. Suppose the positive operator $P_{i}$ corresponds to outcome $0$ and $Q_{i}$ to outcome $1$ for each $\cM_{i}$.

If the outcome of the measurement $\cM_{i}$ is $0$, we simply apply the unitary operation $U_{i}$. If the outcome of the measurement is $1$, we need to ideally perform an inverse operation $(Q_{i})^{-1}$ followed by the $(i+1)^{\rm th}$ POVM $\cM_{i+1}$. Here, we assume that the inverse operations given by $(Q_{i})^{-1}$ only \emph{weakly perturb} the state and hence can be ignored. Thus the $i^{\rm th}$ operation $\cR_{i}$ is characterized by the pair of operators 
\begin{equation}
    R^{0}_{i} = U_{i}P_{i} \equiv K_{i}, \; R^{1}_{i} = Q_{i} \equiv \sqrt{I - K_{i}^{\dagger} K_{i}}. \label{eq:Rmap}
\end{equation} 

We now formally show that the sequence of operations $\cR_{i}$ does implement the map $\cK$ upto $\mathcal{O}(\parallel K^{\dagger}_{i}K_{i} \parallel^{2})$. Note that the deviation from the exact map $\cK$ is thus typically quadratic in the noise strength, as we show numerically for a specific example, in Appendix~\ref{sec:povm_approx}. Using a Taylor series expansion and ignoring all terms of second order and above in $\parallel K^{\dagger}_{i}K_{i} \parallel$, we have, \[ \sqrt{I - K_{i}^{\dagger}K_{i}} \approx I - \frac{1}{2}K_{i}^{\dagger}K_{i}. \] 
The action of the operator $Q_{i}$ is therefore,
\begin{equation}
\label{eq:i}
    Q_{i}\rho Q_{i}^{\dagger} = \rho - \frac{1}{2}K_{i}^{\dagger}K_{i} \rho - \frac{1}{2}\rho K_{i}^{\dagger}K_{i} + \mathcal{O}(\parallel K^{\dagger}_{i}K_{i} \parallel^{2}).
\end{equation}

The action of $\cR_{j}$ on Eq.~\eqref{eq:i}, upto first order in $\parallel K_{i}^{\dagger}K_{i}\parallel$ is,
\begin{align}
\cR_{j}(Q_{i}\rho Q_{i}) &= K_{j}\rho K_{j}^{\dagger} + \rho -  \frac{1}{2}K_{i}^{\dagger}K_{i} \rho 
    - \frac{1}{2}\rho K_{i}^{\dagger}K_{i} \nonumber \\
    & -\frac{1}{2}K_{j}^{\dagger}K_{j} \rho  - \frac{1}{2}\rho K_{j}^{\dagger}K_{j} + \mathcal{O}(\parallel K_{i}K_{i}^{\dagger}\parallel^{2}). \nonumber 
\end{align}
Thus, composing the entire sequence of $N$ maps as in Eq.\eqref{eq:decomp}, we obtain the sum,
\begin{equation}
\begin{aligned}
&& (\cR_{N} \circ \cR_{N-1} \cdots \cR_{1}) (\rho) = \sum_{i =1}^{N}K_{i}\rho K_{i}^{\dagger} + \rho \\ 
    && -  \frac{1}{2}\sum_{i =1}^{N}K_{i}^{\dagger}K_{i} \rho  - \frac{1}{2}\sum_{i =1}^{N}\rho K_{i}^{\dagger}K_{i} + \mathcal{O}(\parallel K^{\dagger}_{i}K_{i} \parallel^{2}). 
\end{aligned}
\label{eq: final_step}
\end{equation}
From the trace-preserving nature of our CPTP map and Eq.\eqref{eq:TP}, we see that  Eq.~\eqref{eq: final_step} simplifies to give us the desired result. 
\end{proof}

\subsection{Quantum circuit for an arbitrary channel via polar decomposition}

We further elucidate the construction outlined in Lemma~\ref{lem:petz_povm} by describing the initial steps of the quantum circuit to realise an arbitrary quantum channel, as shown in Fig.~\ref{fig:blockdiagram}. Each step in our circuit involves implementing a  CPTP map (of rank $2$) with a pair of Kraus operators given in Eq.~\eqref{eq:Rmap}.

For the first step, we implement the operation $\cR_{1}$ whose Kraus operators are given by $K_{1}$ and $\sqrt{I-K_{1}^{\dagger}K_{1}}$. These may be implemented using a single ancilla qubit $A_1$, initialised to the $|{0}\rangle$ state, by first performing an operation $U_{\cM_{1}}$ given by, 
\begin{equation}
\label{equation:UM}
    U_{\cM_{1}} = \begin{bmatrix}
    P_{1} & -Q_{1}\\
   Q_{1} & P_{1}
    \end{bmatrix} .
\end{equation}

Since the operator acts on the combined ancilla-target system given by $\|{0}\rangle\otimes|{\psi}\rangle$ where $|{\psi}\rangle$ is the target state, it is sufficient to define only the first column of the block matrix.

This is then followed by conditionally applying either of two unitary operations -- $U_{1}$ or $\tilde{U}_{1}$ -- on the target state that corresponds to the polar decompositions of the Kraus operators $K_{1}$ and $\tilde{K_{1}} = \sqrt{I-K_{1}^{\dagger}K_{1}}$. We can see that for a rank $N=2$ CPTP map, this gives a complete implementation of the map, as follows. Consider the combined state of the ancilla-target system after the operation of $U_{\cM_{1}}$ and the corresponding unitaries. 
\begin{equation}
    |{\phi}\rangle = |{0}\rangle\otimes U_{1}\sqrt{K_{1}^{\dagger}K_{1}}|{\psi}\rangle + |{1}\rangle\otimes \tilde{U}_{1}\sqrt{\tilde{K_{1}}^{\dagger}\tilde{K_{1}}}|{\psi}\rangle.
\end{equation}
Using Eq.~\eqref{eq:polar} we see that $U_{1}\sqrt{K_{1}^{\dagger}K_{1}} = K_{1}$ and $\tilde{U}_{1}\sqrt{\tilde{K_{1}}^{\dagger}\tilde{K_{1}}} = \tilde{K_{1}}$. Therefore, we me rewrite $|{\phi}\rangle$ as
\begin{equation}
    |{\phi}\rangle = |{0}\rangle\otimes K_{1}|{\psi}\rangle + |{1}\rangle\otimes \tilde{K_{1}}|{\psi}\rangle.
\end{equation}

For channels with $N>2$, we now wish to apply the subsequent set of operations only on the part of the state corresponding to $\tilde{K_{1}}$. To achieve this, introduce a second ancilla qubit ($A_2$) and apply a \textsc{cnot} gate with $A_1$ as the control. $A_1$ can now be reinitialised to the $|{0}\rangle$ state for further operations. The state of the system ($\rho_{A_{1}}\otimes\rho_{A_{2}}\otimes\rho_{T}$) after the \textsc{cnot} and re-initialization is given as,
\begin{equation}
\begin{aligned}
    \rho_{A_{1}}\otimes\rho_{A_{2}}\otimes\rho_{T} = |{0}\rangle\langle{0}| \otimes |0\rangle\langle{0}|\otimes K_{1}|{\psi}\rangle\langle{\psi}|K_{1}^{\dagger} \\
    + |{1}\rangle\langle{1}| \otimes |{0}\rangle\langle{0}|\otimes \tilde{K_{1}}|{\psi}\rangle\langle{\psi}|\tilde{K_{1}}^{\dagger}.
\end{aligned}
\end{equation}

Now, we can use the ancilla system $A_{1}$ as a control to apply the second set of Kraus operators corresponding to $\cR_{2}$.  
We redefine the $U_{\cM_{i}}$ operators for $i>1$ as follows 
\begin{equation}
    \label{equation:UM2}
    U_{\cM_{i}} = \begin{bmatrix}
   Q_{i} & - P_{i} \\
  P_{i} & Q_{i}
    \end{bmatrix} \quad \forall i>1.
\end{equation}
Using this definition of $U_{\cM_{2}}$, the state of the combined system post the re-initialization step is given by 
 \begin{equation}
    \begin{aligned}
    \rho_{a_{2}}\otimes\rho_{a_{1}}\otimes\rho_{s} = |{0}\rangle\langle{0}| \otimes |{0}\rangle\langle{0}|\otimes K_{1}|{\psi}\rangle\langle{\psi}|K_{1}^{\dagger} \\ +  |{0}\rangle\langle{0}| \otimes |{0}\rangle\langle{0}|\otimes K_{2}\tilde{K_{1}}|{\psi}\rangle\langle{\psi}|\tilde{K_{1}}^{\dagger}K_{2}^{\dagger} \\ + 
    |{1}\rangle\langle{1} |\otimes |{0}\rangle\langle{0}|\otimes \tilde{K_{2}}\tilde{K_{1}}|{\psi}\rangle\langle{\psi}|\tilde{K_{1}}^{\dagger}\tilde{K_{2}}^{\dagger}.
    \end{aligned}
\end{equation}

Now, the subsequent steps can be carried out in the same way, and the second control ancilla is  traced over at the very end. Finally, it is now straightforward to apply this circuit construction to obtain an \emph{approximate} circuit implementation of the Petz map, as described in Sec.~\ref{sec:petz results}.

\section{Petz recovery circuit via block encoding}\label{sec:petz_qsvt}

In this section, we present an alternate approach to realizing the Petz map recovery circuit for a $n$-qubit code, using both the techniques of block encoding and isometric extension. We first visualize the code-specific Petz map as a sequence of three completely positive maps, namely, 
\begin{align}
   & \cE(P)^{-1/2}(.)\cE(P)^{-1/2} \label{eq:2}\\
    &\quad \cE^{\dagger} (.) \label{eq:3}\\
    &\quad P\,(.) P \label{eq:4}.
\end{align}

In contrast to the procedure outlined in \cite{gilyen2022_petz}, which requires the block encoding of both the maps in Eqs.~\eqref{eq:2} and \eqref{eq:4}, we only employ the block encoding technique to implement the \emph{normalization} map in Eq.~\eqref{eq:2}. In a further departure from~\cite{gilyen2022_petz} where the quantum singular value transform (QSVT) is used to implement the block encoding of the pseudo-inverse operation $\cE(P)^{-1/2}$, here, we use a more direct approach to construct the block encoding of this operator, as discussed in Appendix~\ref{sec:Appendix B}. 

The adjoint map in Eq.~\eqref{eq:3} and the action of the projection map in Eq.~\eqref{eq:4} are both realised via isometric extensions. If $\cV_{\cE}$ denotes the isometric extension of the quantum channel $\cE$, the isometric extension of the adjoint map $\cE^{\dag}$ on a state $\rho$ can be realised by the action of the unitary $(I_{2}\otimes \cV^{\dag}_{\cE})$ ( $I_{2}$ being the $2\times 2$ identity matrix) as~\cite{2011arXiv1106.1445W},
\begin{align}
    \cE^{\dag}(\rho) &= \langle{0}|(I_{2} \otimes \cV^{\dag}_{\cE}) (I_{2}\otimes \rho)(I_{2} \otimes \cV_{\cE})|{0}\rangle.
\end{align}

Finally, we realise the  projection onto the codespace via the isometric extension $U_{\cP}$ of the channel $\cP$, with Kraus operators 
\begin{equation}
    E^{(1)}_P =\begin{pmatrix}
        P &0 \\
        0 & 0
    \end{pmatrix}, \; E^{(2)}_{P} =\begin{pmatrix}
        P_{\perp}&0 \\
        0 & I_{2^{n+1}}
    \end{pmatrix}, \label{eq:channelP}
\end{equation} 
    where $P_{\perp}$ is the projector orthogonal to the codespace projector $P$. Applying the isometric extension of the channel $\cP$ allows us to post-select for the recovered state in the codespace. The above discussion is formally stated and proved in Theorem~\ref{thm:W petz unitary}.

\begin{theorem}\label{thm:W petz unitary}
   The circuit in Fig.~\ref{fig:petz_reco_state} implements the code-specific Petz recovery map with probability $\frac{1}{N^n||(\cE(\cP)^{-1/2})||^2}$ for an $n-$qubit code with projector $P$, where each qubit is subject to a noise channel $\cE$ with $N<4$ Kraus operators in an i.i.d. fashion. For a Petz channel with $N^n$ Kraus operators, the ancillary inputs are initialized to the state $(|{0}\rangle\langle{0}| \otimes \frac{I_{N^n}}{N^n})$. The final recovered state is obtained by post-selecting the all-zero state on the first $(n\log_{2}N+2)$ qubits. 
\end{theorem}

\begin{figure}[t]
     \centering
  \includegraphics[scale=0.7]{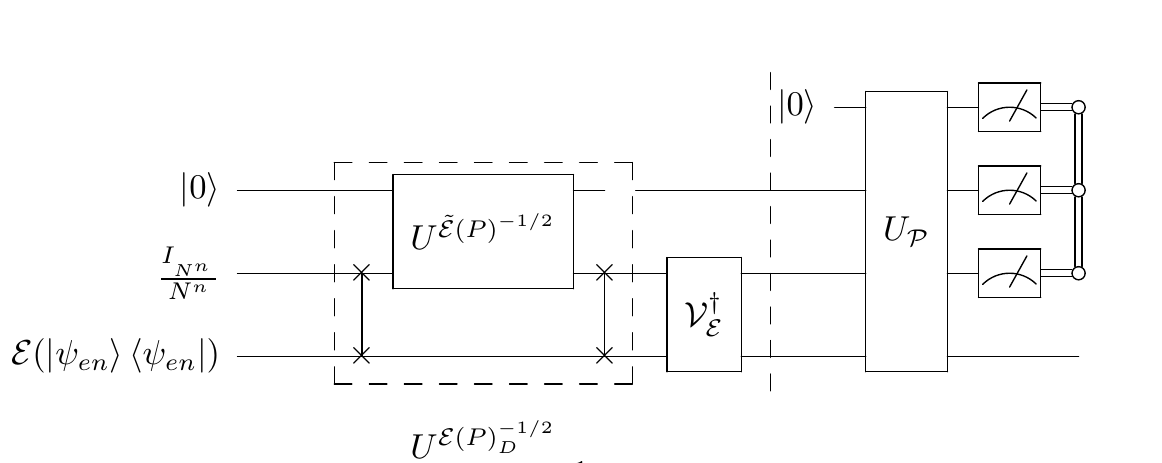}
         \caption{Petz recovery circuit for noise channel $\cE$, using block encoding of $\cE(P)^{-1/2}$ with isometric extensions $V_{\cE}$ and $U_{P}$. $\cE(|\psi_{\rm en}\rangle\langle\psi_{\rm en}|)$ is the noisy, encoded state.} 
         \label{fig:petz_reco_state}
    
     \end{figure}
     
\begin{proof}
 For a $n-$qubit code and an error channel with $N$ Kraus operators, the Petz recovery map $\cR_{P,\cE}$ has $N^n$ Kraus operators. Let $W$ denote the unitary operation implemented by the first part of the circuit in Fig.~\ref{fig:petz_reco_state}, up to the dotted line. 

 \begin{widetext}
\begin{center}
     \begin{figure}[t]
     \centering                                    
\includegraphics[scale=0.77]{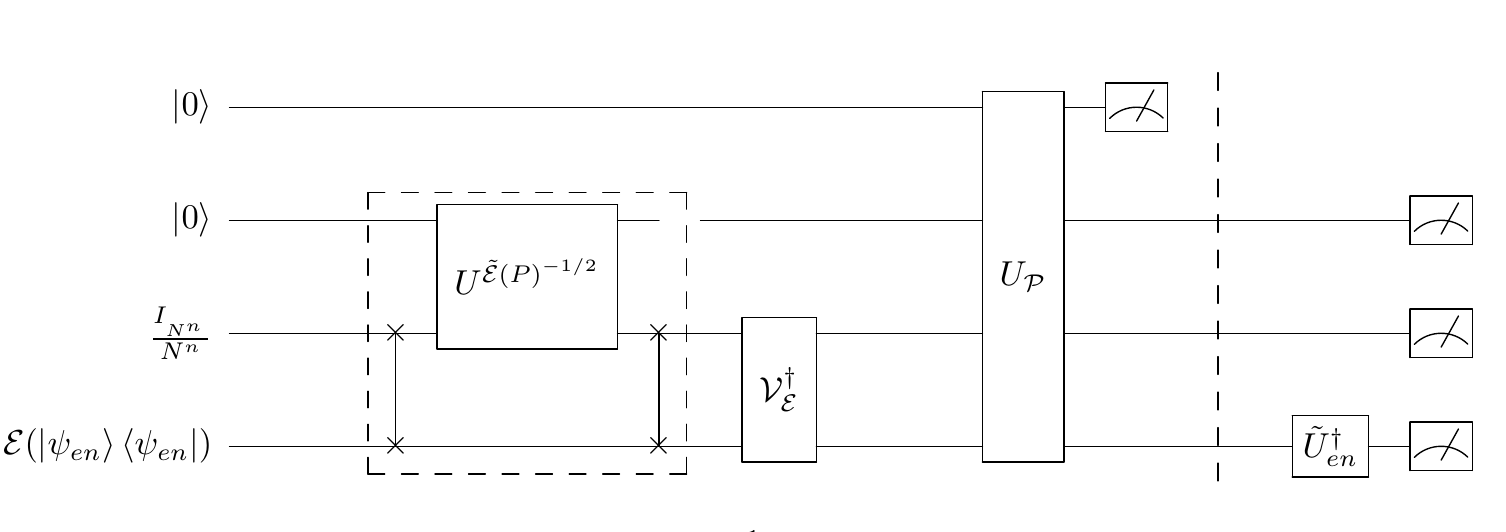}
\caption{Petz recovery circuit followed by the circuit for estimating the fidelity between the encoded state and the recovered state. The third register is an $(n\log_2N)$-qubit state for an $n$-qubit state and a noise channel $\cE$ with $N$ Kraus operators.}
         \label{fig:petz_reco_fid}
    
\end{figure}
\end{center}
\end{widetext}

 We first note that the unitary $W$ has the following algebraic form.

\begin{multline}
    W = (I_{2}\otimes SWAP) (U^{\Tilde{\cE}(P)^{-1/2}} \otimes I_{N^n}) (I_{2}\otimes SWAP)\\
               (I_{2} \otimes \cV_{\cE}^{\dag}) \label{eq:W_matrix}.
\end{multline}
where, $U^{\Tilde{\cE}(P)^{-1/2}}$ denotes the block encoding of the operator $\frac{\cE(\cP)^{-1/2}}{||\cE(\cP)^{-1/2}||}$ and $\cV_{\cE}$ denotes the isometric extension of the noise channel $\cE$. $I_{N^{n}}$ denotes the $N^{n} \times N^{n}$ identity matrix.

Next, we note that if $U^{A}$ gives the block encoding of an operator $A$, then sandwiching $U^A$ with a pair of \textsc{swap} operations effectively gives a block encoding of the block diagonal operator $A_D$, written as, 
\begin{align*}
    A_D & = \begin{pmatrix}
              A & 0 & \hdots &0\\
              0 & A & \hdots & 0\\
              0 & 0 & \ddots & 0\\
              0 & 0 &\hdots & A 
    \end{pmatrix}.
\end{align*} 
We refer to Lemma~\ref{lem: UAd_proof} in Appendix~\ref{sec:Appendix B} for a proof. 

Thus, the sequence of gates $(I_{2}\otimes SWAP) (U^{\Tilde{\cE}(P)^{-1/2}} \otimes I_{N^n}) (I_{2}\otimes SWAP)$ in Eq.~\eqref{eq:W_matrix} gives a block encoding of the block diagonal operator $(\Tilde{\cE}(P)^{-1/2})_D$. Note that the number of diagonal blocks in the matrix $(\Tilde{\cE}(P)^{-1/2})_D$ depends on the number of Kraus operators characterizing the Petz map. For an $n$-qubit code, there exist $N^n$ Kraus operators, each of dimension $2^n \times 2^n$. Thus, $(\Tilde{\cE}(P)^{-1/2})_D$ contain $N^n$ diagonal blocks of the same dimension.

The full sequence of operations in Eq.~\eqref{eq:W_matrix} thus corresponds to sequentially applying the unitaries  $U^{\Tilde{\cE}(P)^{-1/2}_D}$ and $(I_{2}\otimes \cV^{\dag}_{\cE})$. Multiplying these unitaries, we then obtain 
the elements of $i^{\rm th}$ block of the first row of the unitary $W$ as,
\begin{align}
 \frac{ E_i^{\dag} \cE(P)^{-1/2} }{||\cE(P)^{-1/2}||} .
\end{align}

It follows that the action of the unitary operation $W $ on the state $(|{0}\rangle\langle{0}| \otimes \frac{I_{N^n}}{N^n} \otimes \cE(|\psi_{\rm en}\rangle\langle{\psi_{\rm en}}|))$, gives, 
\begin{align}\label{eq:w_in}
  W \left(|{0}\rangle\langle{0}| \otimes \frac{I_{N^n}}{N^n} \otimes \cE(|{\psi_{\rm en}}\rangle\langle{\psi_{\rm en}}|) \right)W^{\dag} \mbox{~~~~~~~~~~~~~}\nonumber\\
  = |{0}\rangle\langle{0}|^{2n+1}  \otimes \Bigg(\frac{(\cE^{\dag}\circ\cE(P)^{-1/2} \circ\cE)(|{\psi_{\rm en}}\rangle\langle{\psi_{\rm en}}|) }{N^n  \,\,||(\cE(\cP)^{-1/2})||^2}\Bigg)\nonumber\\ +\hdots\mbox{~~~~~~~}
\end{align}
The final step (beyond the dotted line in Fig.~\ref{fig:petz_reco_state}) is simply a projection onto the codespace. To implement this, the isometric extension $U_{\cP}$ of the channel $\cP$ defined in Eq.~\eqref{eq:channelP} is applied to the state 
\[ |{0}\rangle\langle{0}|\otimes  W \left(|{0}\rangle\langle{0}| \otimes \frac{I_{N^n}}{N^n} \otimes \cE(|{\psi_{\rm en}}\rangle\langle{\psi_{\rm en}}|) \right)W^{\dag},\] 
evaluated in Eq.\eqref{eq:w_in}. The action of $U_{\cP}$ leads to,
\begin{multline}\label{eq:Ucp_out}
U_{\cP} ( I_2 \otimes W) \Bigg(|{0}\rangle\langle{0}|^{\otimes 2} \otimes \frac{I_{N^n}}{N^n} \otimes \cE(|{\psi_{\rm en}}\rangle\langle{\psi_{\rm en}}|) \Bigg)\\(I_2 \otimes W^{\dag}) U_{\cP}^{\dag}\\
       = |{0}\rangle\langle{0}\otimes|{0}\rangle\langle{0}|^{2n+1} | \otimes \Bigg(\frac{(\cR_{P,\cE}\circ\cE)(|{\psi_{\rm en}}\rangle\langle{\psi_{\rm en}}) }{N^n  \,\,||(\cE(\cP)^{-1/2})||^2} \Bigg) +\\ |{1}\rangle\langle{1}\otimes |{0}\rangle\langle{0}^{2n+1}| \otimes P_{\perp}(*)P_{\perp}.
\end{multline}
Therefore, upon measuring the first $(n\log_{2} N+2)$ qubits -- which constitute the ancilla qubits for the circuit shown in Fig.~\ref{fig:petz_reco_state} -- in the computational basis and keeping only the all-zero outcome gives the recovered state as $\frac{(\cR_{P,\cE}\circ\cE)(|{\psi_{\rm en}}\rangle\langle{\psi_{\rm en}}|) }{N^n  \,\,||(\cE(\cP)^{-1/2})||^2}$. It is also easy to see from Eq.~\eqref{eq:Ucp_out} that this post-selection on the $|{0}\rangle^{\otimes (n\log_{2}N+2)}$ state succeeds with the probability,
\begin{align}
  \mbox{Prob}(0)&= \frac{\tr(\cR_{P,\cE}\circ\cE)(|{\psi_{\rm en}}\rangle\langle{\psi_{\rm en}}|)}{N^n||(\cE(\cP)^{-1/2})||^2}\nonumber\\
   &=\frac{1}{N^n||(\cE(\cP)^{-1/2})||^2} . 
\end{align}
On the other hand, the success probability of the \emph{QSVT} based approach is $\frac{1}{16 N^n||(\cE(\cP)^{-1})||}$, which is 16 times smaller than the success probability of our approach.
\end{proof}

As a simple corollary to the above result, we show how we can also estimate the fidelity between the recovered state $\cR_{P, \cE}\circ\cE (|{\psi_{\rm en}}\rangle\langle{\psi_{\rm en}}|)$ and the encoded state $|{\psi_{\rm en}}\rangle$ using our approach.

\begin{corollary}\label{thm:W petz unitary_fidelity}
The circuit in Fig.~\ref{fig:petz_reco_fid} gives the fidelity between the encoded state $|{\psi_{en}}\rangle$ and the recovered state $(\cR_{P,\cE}\circ\cE)(|{\psi_{en}\rangle\langle{_{en}}}|)$ as the probability of getting the all-zero state. The final step (shown beyond the dashed line) involves the unitary $\tilde{U}_{\rm en}$ which transforms $|{0}\rangle^{\otimes n}$ to the encoded state $|{\psi_{en}}\rangle$, followed by a measurement in the computational basis. 
\end{corollary}

\begin{figure}[t]
    \centering
    \includegraphics[width = 0.65 \columnwidth]{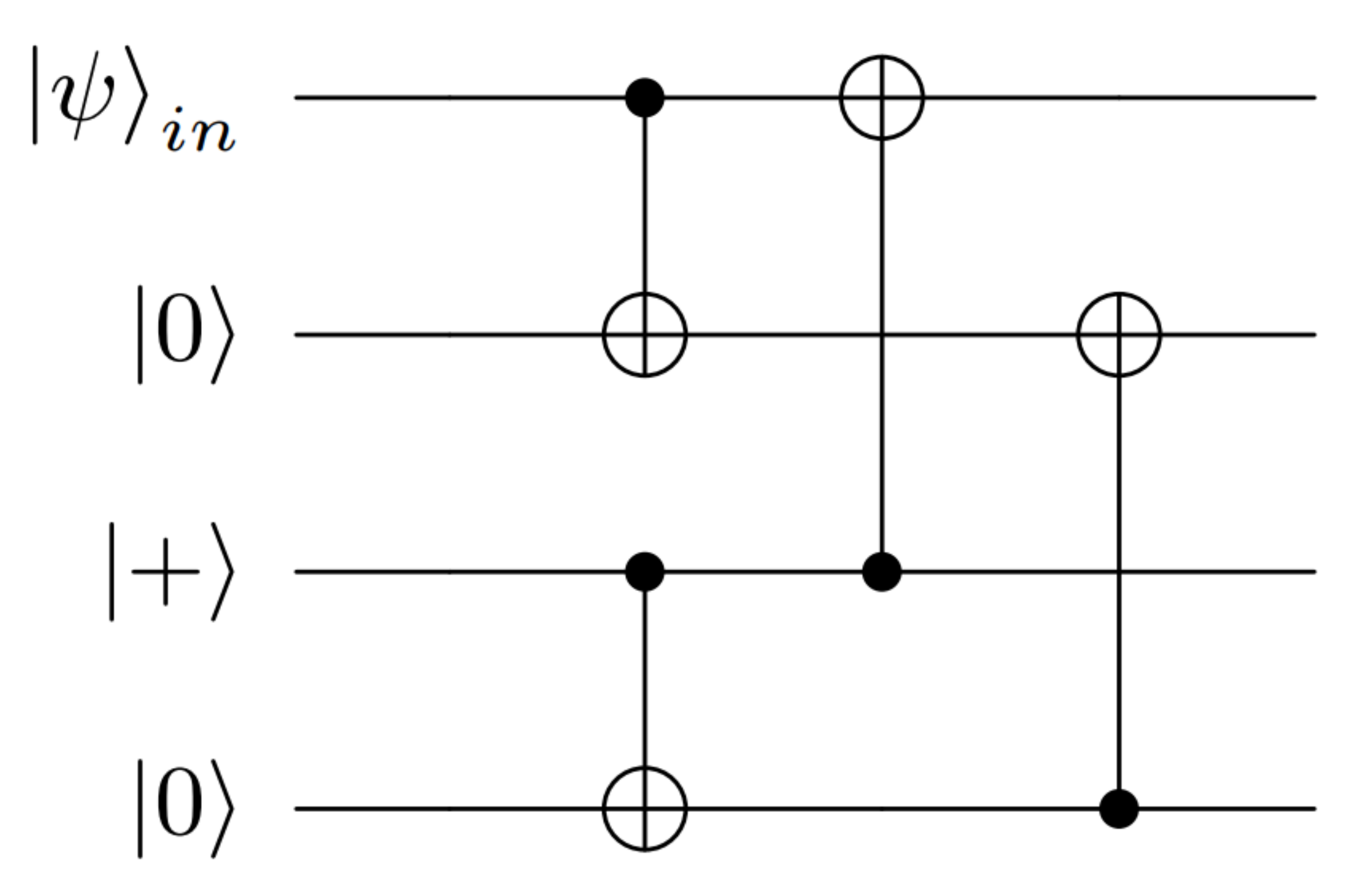}
    \caption{Encoding circuit for the $[[4,1]]$ code.} 
    \label{fig:circ_enco}
\end{figure}

\begin{proof}
Note that the first part (until the dashed line) of the circuit in Fig.~\ref{fig:petz_reco_fid} is the same as that of the circuit in Fig.~\ref{fig:petz_reco_state}. To analyze the last part of the circuit (beyond the dashed line), recall that $\tilde{U}_{\rm en}$ is such that $\tilde{U}_{\rm en}|{0}\rangle^{\otimes n} = |\psi_{\rm en}\rangle$.  Thus, the action of the unitary operator $(I_{2}\otimes I_{N^n}\otimes \Tilde{U}_{en}) U_{\cP} ( I_2 \otimes W) $ on the state $\left(|{0}\rangle\langle{0}| \otimes \frac{I_{N^n}}{N^n} \otimes \cE(|{\psi_{\rm en}}\rangle\langle{\psi_{\rm en}}|)\right)$, results in the following state 
   
    \begin{align}
     &  \left( \frac{\langle{\psi_{en}}(\cR_{P,\cE}\circ\cE)(|{\psi_{en}}\rangle\langle{\psi_{en}}|)|{\psi_{en}}\rangle }{N^n \, ||\cE(\cP)^{-1/2})||^2} \right) |{0}\rangle\langle{0}|^{2n+1} \nonumber \\
     & \qquad\qquad \qquad\qquad\qquad\qquad +  \hdots \label{eq:UW_state} 
    \end{align}
Measuring the state in Eq.\eqref{eq:UW_state} in the computational basis, we see that the probability of obtaining the all-zero state gives the desired fidelity. 
\end{proof}
Finally, in terms of resource requirements, we note that while the QSVT procedure in \cite{gilyen2022_petz} requires $2(n\log_2N+2)$ ancilla qubits, our approach requires only  $(n\log_{2} N +2)$ ancilla qubits. On the other hand, the direct implementation of the Petz map by its isometric extension costs only $n$ ancilla qubits.

\begin{figure}[t]
    \centering
      \includegraphics[width = 0.65 \columnwidth]{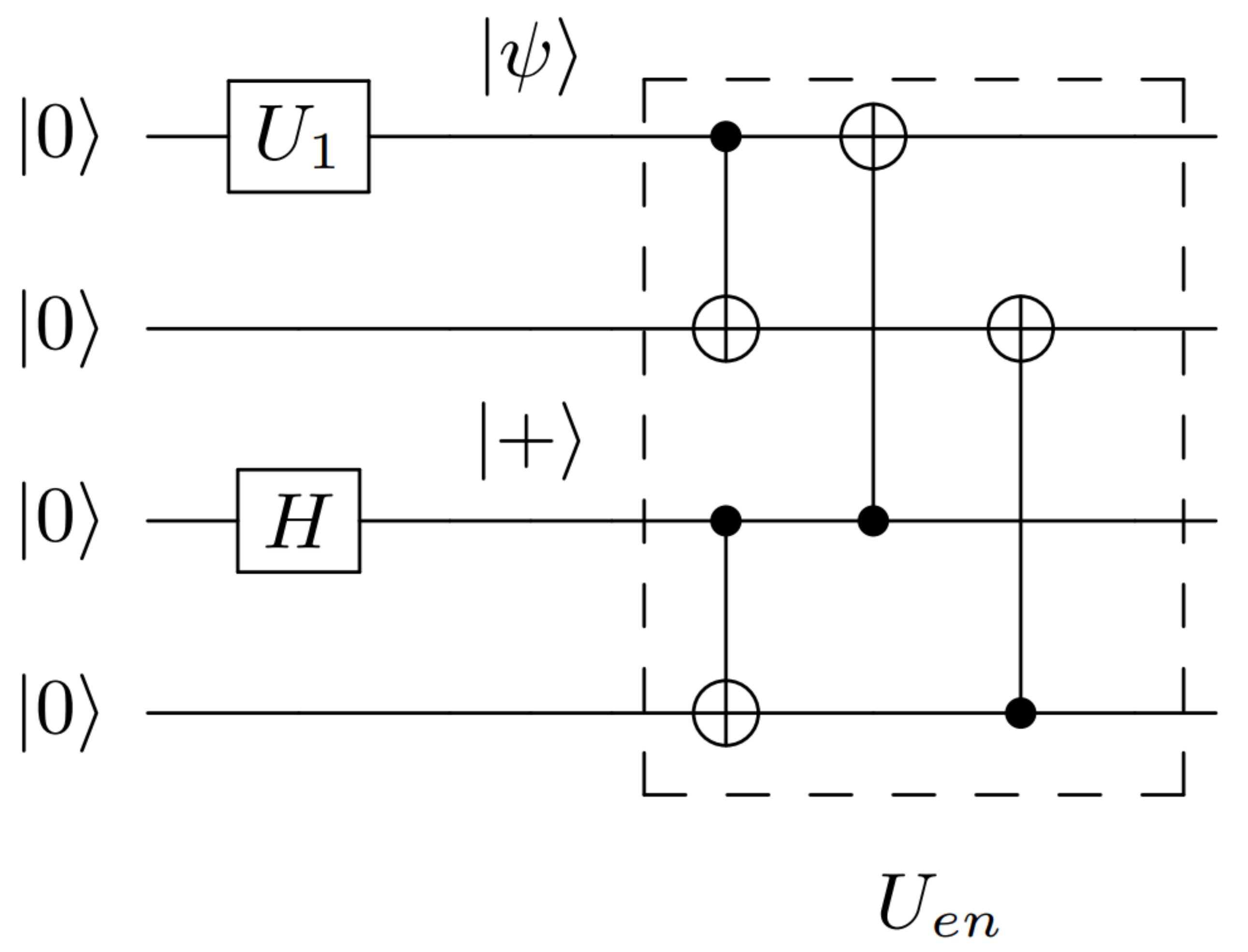}
    \caption{Circuit for the $\Tilde{U}_{\rm en}$ for the $[[4,1]]$ code with encoding unitary $U_{\rm en}$. $U_1$ is the single qubit unitary that prepares an arbitrary single-qubit state $|{\psi}\rangle=a|{0}\rangle+b |{1}\rangle$ from $|{0}\rangle$}.
    \label{fig:Utilde}
\end{figure}

\section{Petz recovery circuits for amplitude-damping noise}\label{sec:petz results}

In this section, we provide concrete examples of our constructions for the case of the $4$-qubit code~\cite{leung} subject to amplitude-damping noise. The amplitude-damping channel models energy dissipation in a quantum systems~\cite{nielsen} and is known to be one of the dominant noise processes in several physical realizations of qubits~\cite{chirolli2008}, making this as a natural candidate for our study.

A single qubit amplitude-damping channel is represented by a pair of Kraus operators, given by,
 \begin{align}
     A_0 = \begin{pmatrix}1 & 0\\
                        0 & \sqrt{1-\gamma}\end{pmatrix} \qquad  A_1 = \begin{pmatrix}0 & \sqrt{\gamma}\\
                        0 & 0\end{pmatrix} ,
 \end{align}
where $\gamma$ is the probability for the system to decay from the (excited) state $|{1}\rangle$ to the (ground) state $|{0}\rangle$. In our simulations, we implement amplitude-damping noise in two different ways. First, we use the circuit model for the amplitude-damping channel discussed in~\cite{nielsen}. Specifically, the isometric extension of a single-qubit amplitude damping channel can be implemented via a controlled rotation operator followed by a \textsc{cnot}, with an ancilla that is initialised to the fixed state $|0\rangle$~\cite{nielsen}. Alternatively, we can also realise amplitude-damping noise using a sequence of identity gates, where the duration and number of gates are decided based on the $T_1$, $T_2$  times of the specific qubits, as described in Appendix~\ref{sec:sim_AD}.

The $4$-qubit code proposed by Leung et al. \cite{leung}, is one of the earliest examples of a channel-adapted quantum error correcting code. It encodes a single qubit into the two-dimensional subspace spanned by,
  \begin{align}\label{eq:leung_code}
      |{0}\rangle_L &= \frac{1}{\sqrt{2}}\left( |{0000}\rangle+|{1111}\rangle \right) \nonumber \\
      |{1}\rangle_L &= \frac{1}{\sqrt{2}}\left( |{0011}\rangle+|{1100}\rangle \right),
  \end{align}
and was shown to correct for amplitude-damping noise with fidelity $1-5 O(\gamma^{2})$.  Subsequently, it was shown that the code-specific Petz recovery map tailored to this $4$-qubit code and amplitude-damping noise improves upon this fidelity bound~\cite{hkn_pm2010}. 

In our implementation, we prepare the encoded states corresponding to the codewords in Eq.~\eqref{eq:leung_code} using the circuit in Fig.~\ref{fig:circ_enco}. Note that in place of the original encoding circuit in~\cite{leung}, we have used a modified form proposed in~\cite{jayashankar2022adaptive}, so as to obtain a simpler circuit for the encoding unitary $U_{\rm en}$. To calculate the fidelity using the circuits in Figs.\ref{fig:bl_enc_rec_1} and~\ref{fig:petz_reco_fid}, we also need a circuit implementation of the unitary $\Tilde{U}_{en}$. The unitary $\Tilde{U}_{en}$, can be constructed once we know the encoding unitary $U_{en}$, as illustrated in Fig.~\ref{fig:Utilde}, for the case of the $4$-qubit code.

\subsection{Results}\label{sec:plots}

We now present the results based on our implementations of the Petz recovery circuits for the amplitude-damping channel and the four-qubit code on a noisy quantum simulator. Our simulations were carried out on the qiskit platform of IBMQ, simulating noisy superconducting qubits. 

\begin{figure}[t]
\centering
\includegraphics[scale=0.23]{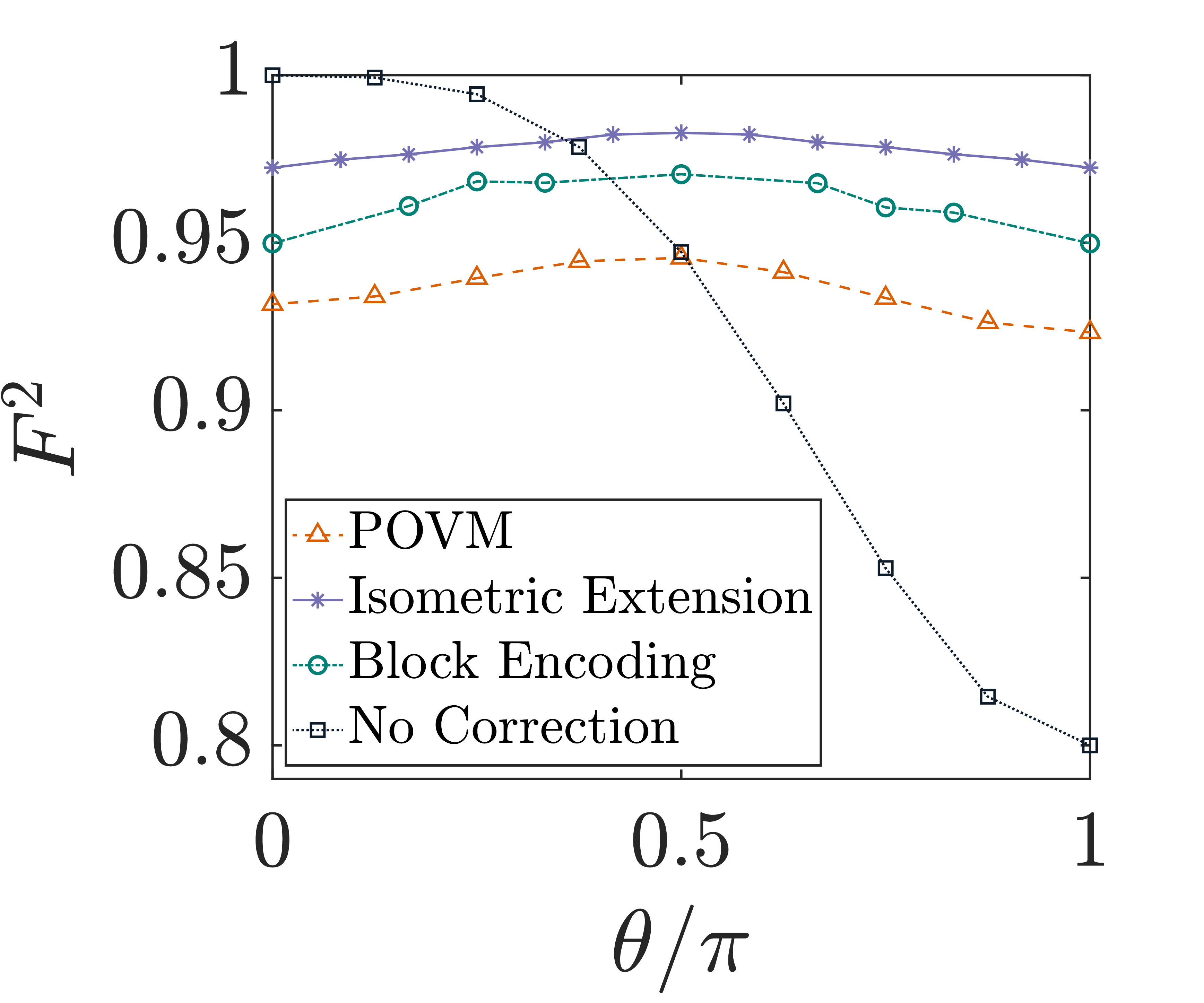}
\caption{Fidelity ($F^2$) for various states encoded as $|{\psi}\rangle = \cos\left(\frac{\theta}{2}\right)|{0}\rangle + \sin\left(\frac{\theta}{2}\right)|{1}\rangle$ in the $4$-qubit code and amplitude-damping noise with damping parameter $\gamma = 0.2$, implemented on the qiskit platform.}
\label{fig:fid_state}
\end{figure}

We first plot the fidelities obtained for different input states for a fixed noise strength, corresponding to a value of the damping parameter $\gamma=0.2$, in Fig.~\ref{fig:fid_state}. Here, we have used the circuit model to introduce amplitude-damping noise for the POVM-based and isometric extension-based Petz circuits. For the Petz circuit using the block encoding approach, we followed the method outlined in Appendix \ref{sec:sim_AD} to introduce the noise. Furthermore, rather than use the block encoding-based circuit in Fig.~\ref{fig:petz_reco_state}, which yields the recovered state with a certain probability, we instead use the circuit in Fig.~\ref{fig:petz_reco_fid} that directly estimates the fidelity of the recovered state. 

Here, we note further challenges in implementing the block encoding-based Petz recovery circuits described in Sec.~\ref{sec:petz_qsvt}, namely, preparing the ancillas in the maximally mixed state. For the specific case of the $4$-qubit code subject to amplitude-damping noise, we require the ancillary qubits to be prepared in the maximally mixed state $\frac{I_{16\times 16}}{2^4}$. Our approach to preparing the maximally mixed state makes use of the noise channel itself. A single qubit amplitude-damping channel with noise parameter $\gamma =0.5$ acting on the state $|{1}\rangle$ results in the maximally mixed state $\frac{I_{2\times2}}{2}$. So for the $4$-qubit maximally mixed state, we simply apply the $4-$qubit amplitude-damping channel to the state $|{1}\rangle^{\otimes 4}$. On the IBMQ quantum simulator, we again follow the procedure outlined the Appendix~\ref{sec:sim_AD} to realize the maximally mixed state $\frac{I_{16\times 16}}{2^4}$.

\begin{figure}[t]
\centering
\includegraphics[scale=0.24]{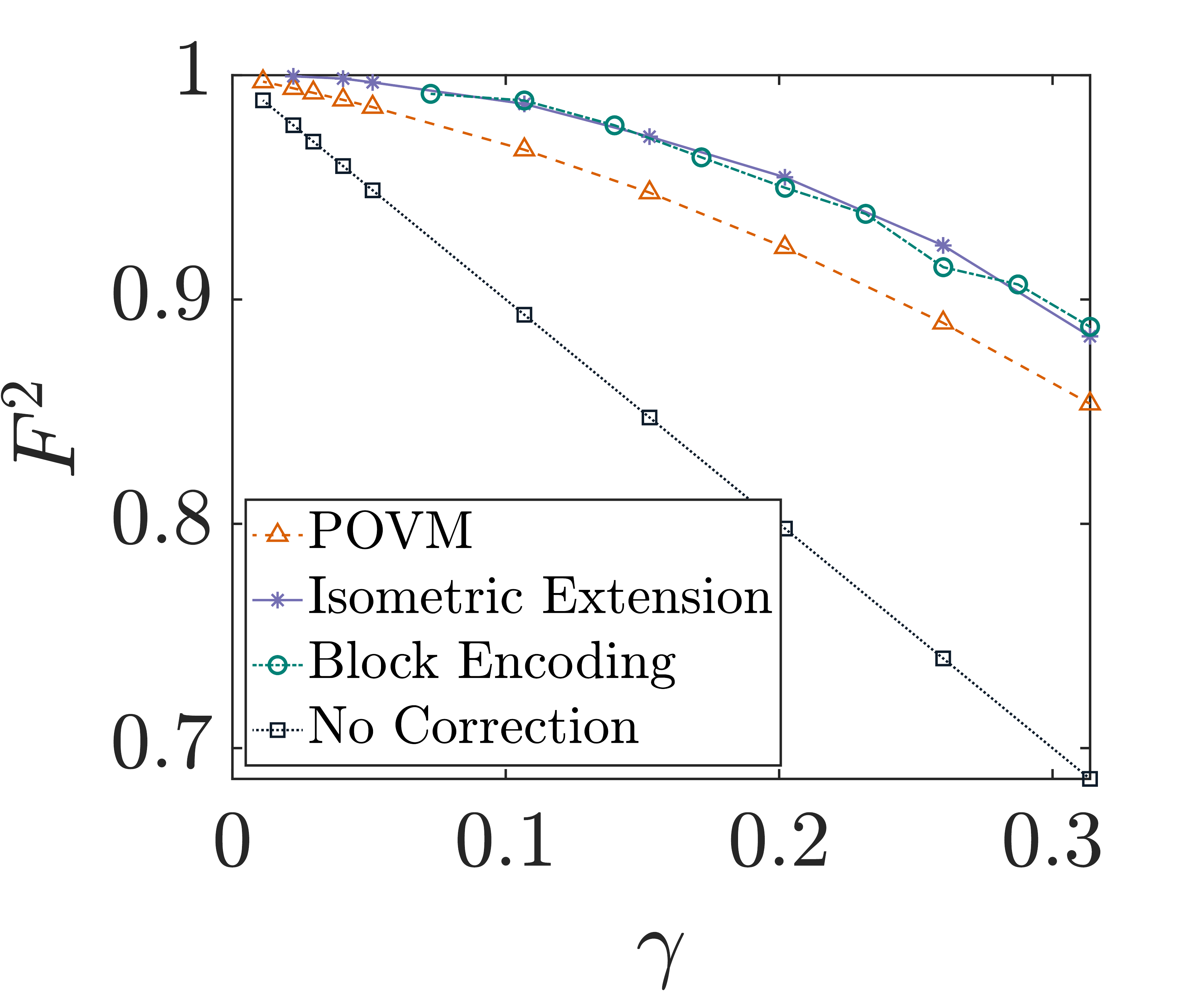}
\caption{Fidelity ($F^2$)  of $|{1}\rangle$ plotted against the damping parameter $\gamma$ for the $4$-qubit code under amplitude-damping noise, simulated using a sequence of identity gates as explained in Appendix~\ref{sec:sim_AD}.}
\label{fig:gamma_vs_fid_ket_1}
\end{figure}

Next, we plot the fidelities for a fixed state, namely, the $|{1}\rangle$ state, for various values of the damping parameter $\gamma$ in Fig.~\ref{fig:gamma_vs_fid_ket_1}. Here we have introduced the noise using the sequence of identity gates as explained in Appendix ~\ref{sec:sim_AD}, for all three Petz circuits.
In this context, it is important to note that the fidelity for the POVM method is slightly lower than the other two constructions that use the isometric extension and sequential block encoding, since the former is an \emph{approximate} realisation of the Petz recovery, rather than an exact implementation.  

\begin{figure}[ht!]
    \includegraphics[width = 0.8 \columnwidth]{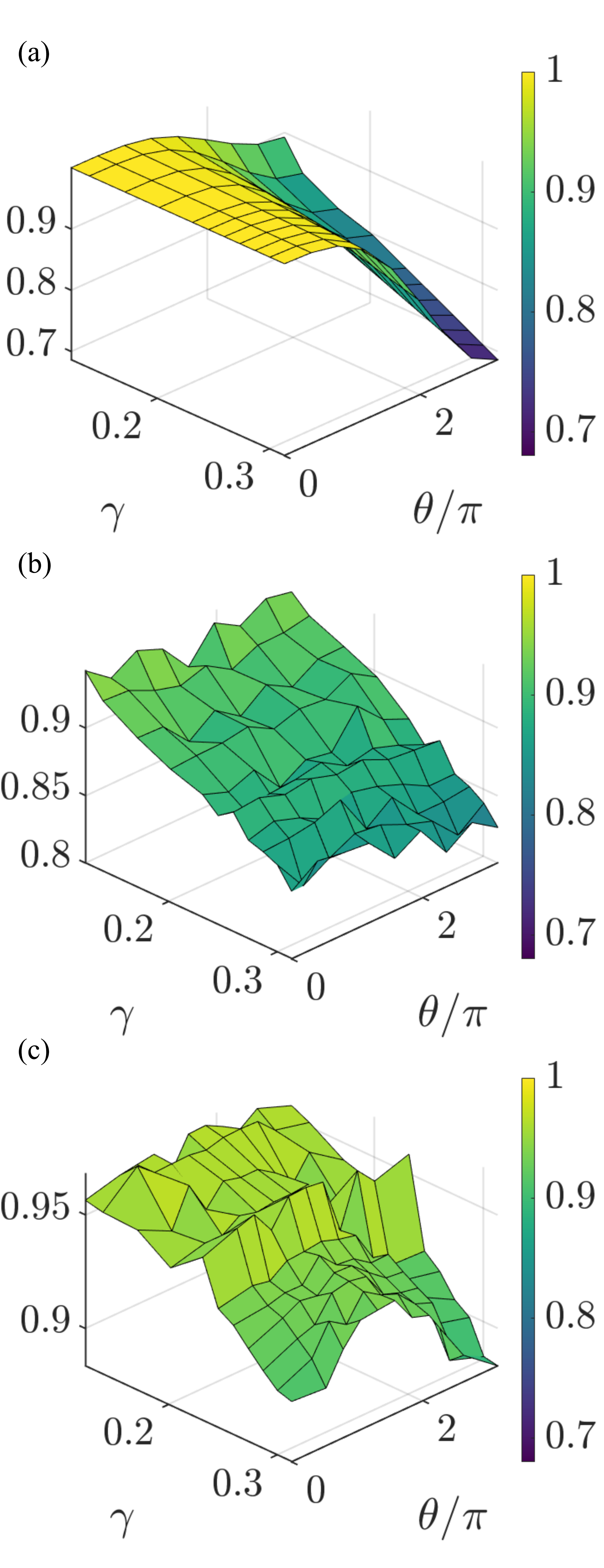}
    \caption{Fidelity maps for (a) an unencoded qubit, (b) error-corrected qubit using the $4$-qubit code and the POVM-based Petz recovery circuit, and, (c) error-corrected qubit using the $4$-qubit code and isometric-extension based Petz recovery circuit. In all three cases, the qubits undergo amplitude-damping noise, simulated using a sequence of identity gates as explained in Appendix~\ref{sec:sim_AD}. The \textsc{cnot} and single-qubit gates used in (b) and (c) are noisy with noise parameters  $\mu = 10^{-6}$ and  $\mu = 10^{-5}$ respectively.}
    \label{fig:fidmap}
\end{figure}

Finally, to emphasize the efficacy of using the Petz recovery channel even in the presence of noisy gates, we simulate our recovery circuits using noisy quantum gates. Unlike the simulation results depicted in Fig.\ref{fig:fid_state} and Fig.\ref{fig:gamma_vs_fid_ket_1}, where the gates are all assumed to be ideal (noiseless), for the simulation results depicted in Fig.\ref{fig:fidmap}, we use noisy gates. Specifically, we assume that the gate-noise in the \textsc{cnot} and single-qubit gates can be modeled as depolarizing noise parameter $\mu$. Figure(\ref{fig:fidmap}) depicts three fidelity maps that allow us to compare the fidelities of the unencoded qubits initialised as $|{\psi}\rangle = \cos\left(\frac{\theta}{2}\right)|{0}\rangle + \sin\left(\frac{\theta}{2}\right)|{1}\rangle$ with the fidelities of the error-corrected qubits, for different values of the damping parameter $\gamma$. The values of $\mu$ for the two different Petz circuits -- one based on isometric extension and the other based on two-outcome POVMs -- have been chosen based on the corresponding threshold values identified in Appendix~\ref{sec:noisy_sim}. 

\subsection{Resource Requirements for Petz recovery circuits}\label{sec:complexity}

With noisy intermediate-scale quantum devices in mind, the objective of any error correcting protocol is to try and optimize the computing resources used. Here, we compare and contrast the resource requirements for the three different recovery circuits outlined in this article. Specifically, we quantify the gate complexity and the number of ancillary qubits required to implement the code-specific Petz map for an $n$-qubit subject to a noise channel with $N = 4$ Kraus operators. 

We first quantify the resource requirements for the Petz circuit implementation in Sec.~\ref{sec:iso petz} based on the isometric extension. It is well known that an $n$-qubit two-level unitary requires $\cO(n^2)$ single-qubit gates and \textsc{cnot} gates \cite{nielsen}. From Lemma~\ref{lem:Petz_isometry}, we note that to implement the $n$-qubit Petz recovery circuit specific to an $n$-qubit code and a noise channel with $N$-Kraus operators via the isometric extension, we require $4^{2n}$ two-level unitaries. Therefore to implement an $n$-qubit Petz recovery channel, we need at most $\cO(n^2 4^{2n})$ single-qubit gates and \textsc{cnot} gates. In addition to this, we also require $\log_2N$ ancillary qubits to implement an $n$-qubit Petz recovery channel.

The POVM-based implementation discussed in Sec.~\ref{sec:petz_povm} requires only $2$ ancillary qubits for any CPTP map of rank greater than two. For a rank-$2$ map, however, it is easy to see that the second ancillary qubit can be dropped, and the circuit can be constructed using just one ancillary qubit. We note that the number of the $U_{M_i}$ matrices is $N^n$, and for each $n+1$-qubit $U_{\cM_i}$, we need two $n$-qubit unitaries $U_i$ and $\Tilde{U}_i$. For implementing each $n$-qubit unitary, we need $\cO((n)^2 4^{n})$ single-qubit and two-qubit \textsc{cnot} gates~\cite{nielsen}. Therefore, for any $n+1-$qubit unitary, we need $\cO((n+1)^2 4^{(n+1)})$ \textsc{cnot} and single-qubit gates. Therefore the total number of gates needed for the implementation of the POVM-based method scales as $\cO( N^n4^n (n^2 +4^{(n+1)}(n+1)^2 ))$, where $N$ is the rank of the noise channel.  

Finally, for the block encoding-based approach in Sec.~\ref{sec:petz_qsvt}, note that we need to perform the singular value decomposition for the $\cE(P)^{-1/2}$. Therefore for an $n-$ qubit $\cE(P)^{-1/2}$, we need to implement the $n$-qubit unitaries $V_1$ and $V_2$ (Appendix~\ref{sec:Appendix B}). Therefore, the total number of the single qubit and the two-qubit \textsc{cnot} gates scales as $\cO(n^2 4^n)$. To realize the unitary $W$, we need to implement the block encoding for $\cE(P)^{-1/2}_D$, which needs $6n\log_2N$ \textsc{cnot} gates. The implementation of the isometric extension for the adjoint channel needs $\cO(n^24^{n}N^n)$ \textsc{cnot} and single-qubit gates. Therefore the total number of gates scales as $\cO(6n\log_2 N+n^2 4^n (1+N^n)) $.  To construct the block encoding for any $n$-qubit operator, we need just one extra ancilla. However, to implement the unitary $W$ in the Eq.\eqref{eq:W_matrix} along with the isometric extension of the channel $\cP$, we need $n\log_2N+2$ ancillary qubits. 

We summarize the resource requirements quantified here in Table~\ref{Tab: Table 1}. Restricting our attention to codes that encode $1$ qubit in $n$, we present a worst-case resource estimation for the different methods, assuming that the single-qubit noise channel has the maximum possible number of Kraus operators, namely, $N=4$.

Thus far, our circuit constructions for the Petz map have relied on knowledge of the Kraus representation of the map. Alternately, we could consider a Lindbladian representation of the Petz recovery channel \cite{Hyukjoon} and design a circuit for the Petz map using its Lindbladians, by following the procedure outlined in the \cite{Lindblad_circ}. Such a circuit would require approximately $\cO(4^{4n})$ single and two-qubit gates, along with $3n$ ancillary qubits, for an $n$-qubit code.

\section{Concluding Remarks}\label{sec:summary}
In this work, we describe three different approaches to obtain circuit implementations of a noise-adapted recovery map, namely the Petz map. Since we focus on the code-specific Petz map here, the resulting recovery circuits can be tailored both to the underlying quantum code as well as the noise channel that the individual qubits are subject to. Apart from the recovery circuits, we also present circuits that can directly estimate the fidelity between the original encoded state and the recovered state.

Of the three approaches discussed here, the circuit construction based on the isometric extension proves to be rather advantageous, both in terms of resourcefulness as well as faithfulness of implementation. Moving away from the standard approach to constructing the isometric extension, we observe here that it suffices to partially decompose the unitary on the extended space, leading to a reduced gate complexity. Furthermore, this approach leads to an \emph{exact} implementation of the Petz recovery map, in contrast to the QSVT-based approach~\cite{gilyen2022_petz}, which realizes the Petz map approximately.

In contrast, the POVM-based approach that we present here is indeed an approximate method to implement the Petz map. However, this approach requires the least number of ancillas, just one in the case of a noise channel with two Kraus operators. We also note here that the isometric extension and the POVM-based approaches to implement the Petz map have an additional advantage over the QSVT-based algorithmic approach -- the latter is probabilistic and requires post-selection, whereas the two approaches presented here are deterministic.

 We also present a third approach which is probabilistic, that uses a combination of the block encoding technique as well as isometric extension. This implementation is comparable to both the isometric extension as well as the recently obtained QSVT-based circuit implementation~\cite{gilyen2022_petz}, in terms of its resource requirements. Like the QSVT-based approach, it is a \emph{probabilistic} implementation, albeit with a higher success probability than that obtained via QSVT. 

Finally, we simulate our circuits on noisy quantum processors to benchmark their performance under both ideal and noisy conditions. For the specific case of amplitude-damping noise, we hence obtain a threshold for \textsc{CNOT} gate-noise below which our recovery circuits can indeed be effective in protecting against the native damping noise of the qubits.
\begin{widetext}
\begin{center}
\begin{table}[t]
    \centering
  \begin{tabular}{|p{2.8cm}||p{3.5cm}|p{2.3cm}|p{2.5cm}|p{2.5cm}|}
\hline
       &\multirow{2}{9em}{\textsc{cnot}s $+$ single-qubit gates}  & Ancilla   &\multirow{2}{7em}{ Approximate }&\multirow{2}{7em}{ Probabilistic } \\
       &                 &                  &                                  & \\
        &                 &                  &                                  & \\
       \hline
        &                 &                  &                                  & \\
 Isometric Extension      &         $\cO(n^2 4^{2n}$) &~~ $2n$        & ~~~No  &~~No \\
  &                 &                  &                                  & \\
  \hline
   &                 &                  &                                  & \\
  POVM      &        $\cO( 4^{2n} (5n^2 + 8n +4))$ & ~~~$2$           & ~~~Yes  & ~~No  \\
   &                 &                  &                                  & \\
       \hline
        &                 &                  &                                  & \\
  Block Encoding     &$\cO( n^{2}4^{2n} + n^2 4^n)$    &  $(2n+2)$                      & ~~~No                                & ~~Yes\\
       \hline
        &                 &                  &                                  & \\
   QSVT~\cite{gilyen2022_petz}    &  $\cO(4^{4n} + n^2 4^n )$    &  $2(2n\,\,+2)$ &~~ Yes &~ Yes\\
    &                 &                  &                                  & \\
   \hline
\end{tabular}
    
     \caption{\label{Tab: Table 1}Resources for different circuit implementations of the Petz recovery map}
\end{table}
\end{center}
\end{widetext}
The techniques outlined in our work are indeed quite general and can be used to obtain resource-efficient circuits for any CPTP map. The fidelity estimation circuits described here might also be of independent interest and find uses beyond the context of QEC. 
Going forward, it will be interesting to study if the noise-adapted recovery circuits described here can be optimised to take into account gate-noise as well as hardware constraints such as qubit connectivity. Our work is thus a first step towards exploring the possibility of implementing noise-adapted recovery maps in a fault-tolerant manner.

\section*{Acknowledgments}
We acknowledge the use of IBM Quantum for this work. The views expressed are those of the authors and do not reflect the official policy or position of IBM or the IBM Quantum team. This research was supported in part by a grant from the Mphasis F1 Foundation to the Centre for Quantum Information, Communication, and Computing (CQuICC). We also acknowledge financial support from the Department of Science and Technology, Government of India, under Grant No. DST/ICPS/QuST/Theme-3/2019/Q59.
\providecommand{\noopsort}[1]{}\providecommand{\singleletter}[1]{#1}%

\appendix

\section{Block encoding of any n-qubit matrix via controlled rotations} \label{sec:Appendix B}
In contrast to the QSVT technique used in \cite{gilyen2022_petz} to construct the block encoding of a function of a matrix, we use an alternate approach based on controlled-rotation gates. Our approach is similar in spirit to the block encoding technique in~\cite{bl_enc}, which, however, makes use of reflection operations rather than rotation gates.
\begin{lemma}\label{lem:block encoding proof}
For any $n$-qubit square matrix $A$ with $||A||\leq 1$, we can construct the exact block encoding via multi-qubit controlled-rotation gates, as shown in Fig.~\ref{fig:blenc_Eph}.  
\end{lemma}

\begin{figure}[t]
    \centering
   \includegraphics[width=1\columnwidth]{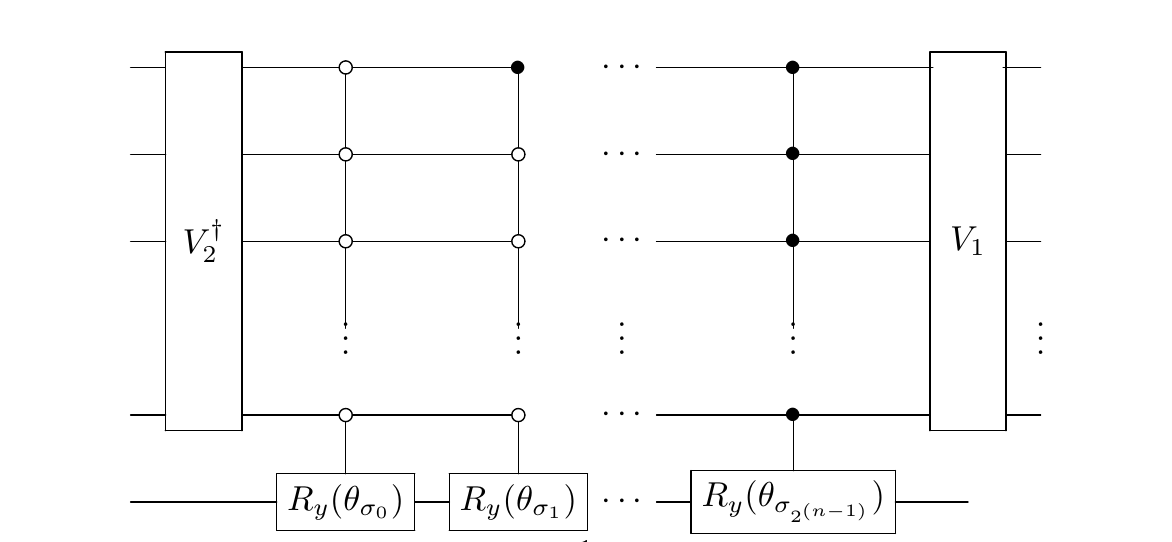}
    \caption{Block encoding of any square matrix $\frac{A}{||A||}$.  }
    \label{fig:blenc_Eph}
\end{figure}

\begin{proof}
For an $n$-qubit square matrix $A$, with $||A||\leq 1$, we consider the singular value decomposition of the matrix $A$, given by, 
\begin{align}
    A &= V_1 \Sigma V_2^{\dag}.
\end{align}
Here, $\sigma$ is a diagonal matrix containing the singular values $\sigma_i$ of $A$, and $V_1$ and $V_2$ are unitary matrices.\\
We first obtain a block encoding of $\Sigma$, via an $(n+1)$-qubit unitary, as, 
\begin{align}\label{eq:Sigma _n+1}
U^{\Sigma}&= \begin{pmatrix}
         [ \Sigma]_{2^n \times 2^n} & & -\sqrt{I_{2^n}-\Sigma^2} \\
         \\
          \sqrt{I_{2^n}-\Sigma^2}& & [ \Sigma]_{2^n \times 2^n}.
\end{pmatrix}
\end{align}
It is now straightforward to see that the block encoding of $\Sigma$ can be implemented in a quantum circuit via multi-controlled rotation gates $R_y(\theta_{\sigma})$, with rotation angles $\theta_{\sigma} = \cos^{-1}{\sigma}$ depending on the singular values, as defined below.
\begin{align*}
     R_y(\theta_{\sigma_i}) & = \begin{pmatrix}
                  \sigma_i & -\sqrt{1-\sigma_i^2}\\
                  \sqrt{1-\sigma_i^2} & \sigma_i
        \end{pmatrix}; \; \theta_{\sigma_{i}} = \cos^{-1}{\sigma_{i}}.
\end{align*}
To achieve the block encoding for the matrix $A$, we then sandwich $U^{\Sigma}$ in between $(I_{2} \otimes V_1)$ and $(I_{ 2} \otimes V_2^{\dag})$ as follows. 
\begin{align}\label{eq:proof of blnc}
  U^A&= (I_{2} \otimes V_1)U^{\Sigma}(I_{2} \otimes V_2)   \nonumber   \\
 \therefore\,\,U^A   &= \begin{pmatrix}
              A& * \\
              * &*
    \end{pmatrix}.
\end{align}
Therefore, the circuit in Fig~\ref{fig:blenc_Eph} implements the $U^A$.
\end{proof}
Note that for a matrix $A$ with $||A||>1$, we can construct the block encoding of $\frac{A}{||A||}$, following the same method as above, by modifying the rotation angles as $\theta_{\sigma} = \cos^{-1}(\frac{\sigma}{||A||})$. Next, we show how the block encoding of $A$ can be used to obtain a block encoding of its block diagonal form $A_{D}$.
\begin{figure}[t]
\centering
    \includegraphics[width = 0.75 \columnwidth]{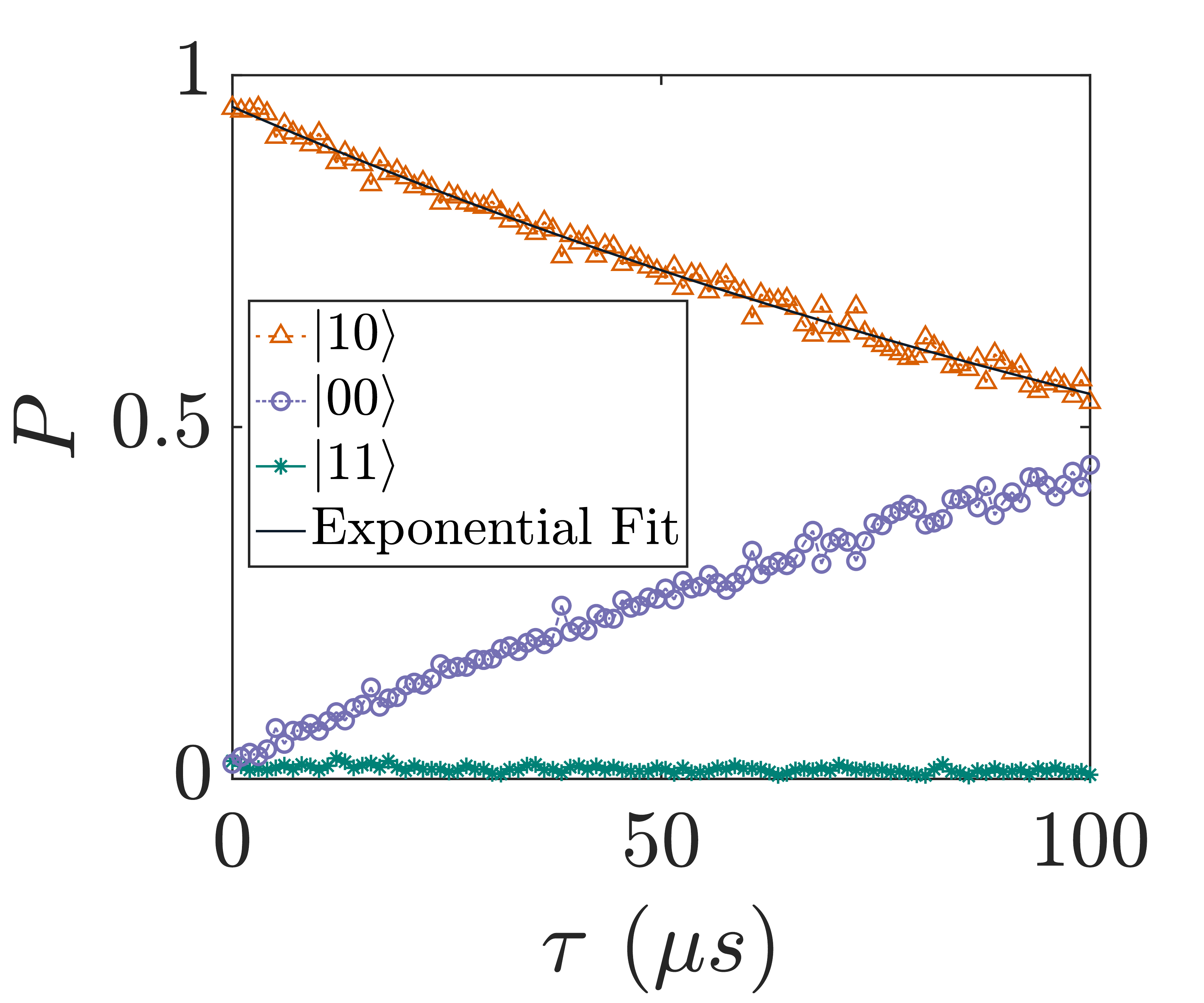}
    \caption{The probability ($P$) of measuring the two-qubit states $|{01}\rangle$, $|{00}\rangle$ and $|{11}\rangle$ after the qubits were initialised in the $|{01}\rangle$ state, and a  delay ($\tau$) was introduced before measurement.}
    \label{fig:experiment}
\end{figure}

\begin{lemma}\label{lem: UAd_proof}
If the $(n+1)$-qubit unitary $U^A$ is a block encoding of an $n-$qubit matrix $A$, then
\begin{align}
 U^A_D &= (I_{2} \otimes SWAP_{mn} ) (U^A \otimes I_{2^m}) (I_{2} \otimes SWAP_{m n})   \label{eq:UAd} .
\end{align}

is a block encoding of the associated block diagonal matrix 
\[ A_{D} = \begin{pmatrix} A& \hdots & 0 \\
                                          0&  \ddots & 0\\
                                          0 &  \hdots & A\end{pmatrix}_{2^{(n+m)} \times 2^{(n+m)}}.\]
                                           
Here, $I_m$ is the $m-$qubit identity matrix and $SWAP_{n m}$ swaps the states of the $n$-qubit system with that of the $m$-qubit system.
\end{lemma}
 \begin{proof}
 Consider the top left block of the $U^A_D$ in Eq.\eqref{eq:UAd}. It is easy to check that this evaluates to,
 
 \begin{align}
      \langle0|\otimes I_{2^{(m+n)}} ) U_D^A (|0\rangle\otimes I_{2^{(m+n)}}) &= SWAP_{m n} (A_n \otimes I_{2^m})\nonumber\\
      &SWAP_{m n} =I_{2^m}\otimes A_n  ,
 \end{align}
as desired.
\end{proof}

\section{Simulating Amplitude-Damping Noise }\label{sec:sim_AD}
In this section, we show how the amplitude-damping channel can simply be realised as a sequence of noisy, identity gates by appropriately tuning the  $T_1$ and $T_2$ relaxations that are commonly used to describe the coherent properties of density matrices. The $T_1$ processes describe the relaxation of the diagonal elements, whereas the $T_2$ processes describe the same for the off-diagonal elements. 

The action of these errors on the density matrix can be given as~\cite{nielsen}
\begin{equation}
    \label{equation:T1}
    \rho \rightarrow \begin{pmatrix} 1-e^{\frac{-t}{T_1}}\rho_{11}&  e^{\frac{-t}{T_2}}\rho_{01} \\
                                          e^{\frac{-t}{T_2}}\rho_{01}^* &e^{\frac{-t}{T_1}}\rho_{11}\end{pmatrix}.
\end{equation}
Comparing the above equation with the effect of the amplitude-damping channel on a density matrix, we can see that for $T_2 = 2*T_1$ we get a decoherence process that exactly mimics an amplitude-damping channel with decay rate $\gamma = 1 - e^{\frac{-t}{T_1}}$. We make use of this observation and set the $T_{1}$, $T_{2}$ times of the qubits in our qiskit simulation so as to ensure that the qubits undergo only amplitude-damping noise. 

Once we fix the $T_{1}$ and $T_{2}$ times in our noisy simulation, we can sweep through different values of the noise strength by introducing sequences of identity gates of varying lengths. Applying a sequence of $N$ identity gates, where each gate is of duration $t$ seconds, results in a qubit that is \emph{idling} for $Nt$ seconds. Such a qubit is then effectively subject to an amplitude-damping channel with damping rate $\gamma$, given by
\begin{equation}\label{eq:gamma}
    \gamma = 1-e^{\frac{N*t}{T_1}}.
\end{equation}
Typical hardware implementations fix the time $t$ of the identity gate as $35ns$. We fix this value and vary the length $N$ of the sequence in order to simulate amplitude-damping channels of different noise strengths $\gamma$. \\
Our approach to simulating amplitude-damping noise on superconducting qubit platforms is justified by Fig.~\ref{fig:experiment}, which shows the result of a simple experiment run on the IBM Mumbai quantum processor involving such a sequence of identity gates. The delay introduced by the sequence of identity gates, was varied between $0$ to $100\mu s$. The plot clearly shows the decay of the $|{01}\rangle$ state to the $|{00}\rangle$ state without any noticeable buildup of the $|{11}\rangle$ state, indicating the dominant presence of amplitude-damping noise. The plot also shows an exponential fit ($R^2=0.9904$) for the decay probability of the $|{01}\rangle$ overlayed over the actual hardware data. 

\begin{figure}[t]
    \centering
    \includegraphics[width=1\columnwidth]{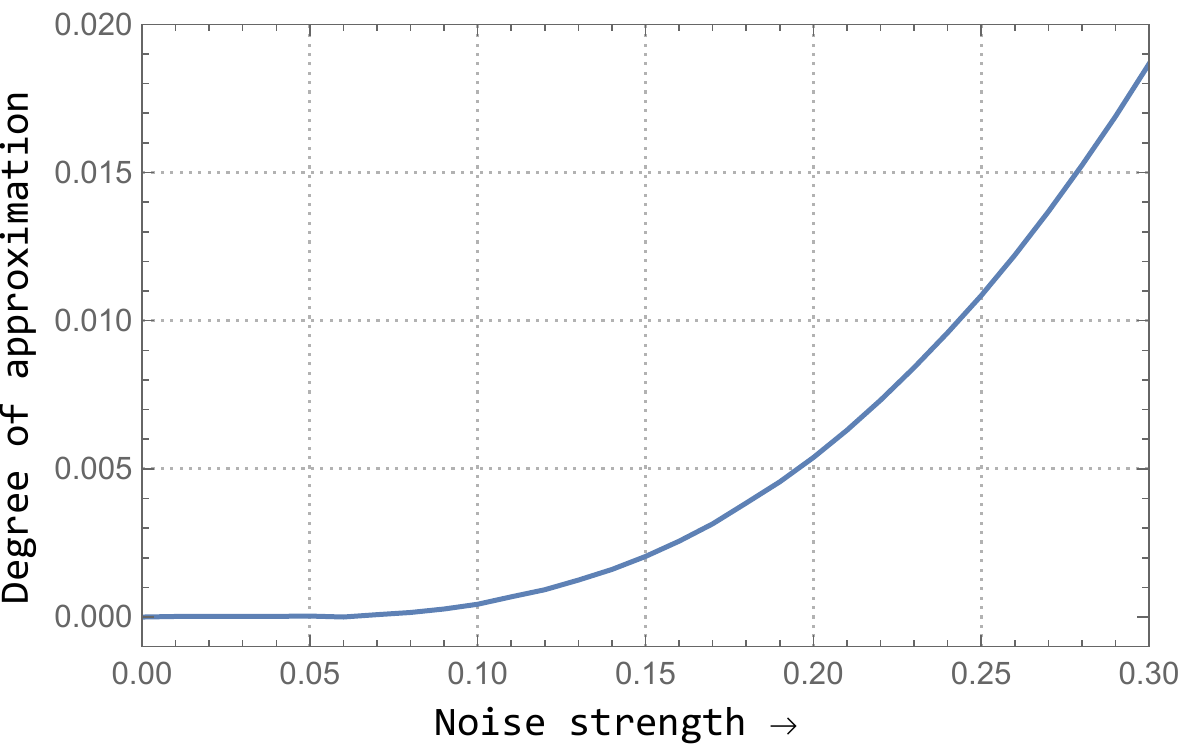}
    \caption{Degree of approximation $\Delta$ as a function of the damping strength $\gamma$.}
    \label{fig:deg_of_app}
\end{figure}

This forms the basis for how we simulate amplitude-damping noise in the simulations presented in Sec.~\ref{sec:petz results}, essentially by making the encoded state \emph{wait} using a sequence of identity gates, before applying the recovery circuits. 

\section{Degree of Approximation for the POVM-based approach}\label{sec:povm_approx}
As discussed in the Sec.\ref{sec:petz_povm}, our circuit construction based on the two-outcome POVM-based decomposition is an approximate method to realize the Petz map (or any CPTP map). Here, we estimate the degree of approximation in terms of the difference in the worst-case fidelity of the actual Petz map and approximated Petz map. The degree of approximation is quantified by, 
\begin{eqnarray}
    \Delta &=& |F^2_{\rm min}(\cC,\Tilde{\cR}_{P,\cE }\circ \cE (\ket{\psi}\bra{\psi})) \nonumber \\
    &&  -  F^2_{\rm min}(\cC,\cR_{P,\cE }\circ \cE (\ket{\psi}\bra{\psi}))|,
\end{eqnarray}
where $\Tilde{R}_{P,\cE}$, is the approximated form of the Petz map. As an example, we compute the degree of approximation $\Delta$ for the code-specific Petz map, specific to the $4$-qubit code in Eq.~\eqref{eq:leung_code}, as a function of the damping strength ($\gamma$). Fig.~\ref{fig:deg_of_app} shows the variation of $\Delta$ as a function of $\gamma$. We note that the mismatch in the performance of the approximated Petz and the actual Petz increases as follows
\begin{align}
    \Delta & =  0.0414\, \gamma^2+ (7\times 10^{-3})\,\gamma + (1.2\times 10^{-4}),
\end{align}
which is effectively quadratic in $\gamma$ as expected from the expresison in Eq.~\eqref{eq:decomp}. 

\section{Noisy Simulation of the Petz Recovery Circuit}\label{sec:noisy_sim}

A majority of the results obtained in Sec.\ref{sec:petz results} are based on ideal simulations using faultless quantum gates. However, the fidelity map plots shown in Fig.~\ref{fig:fidmap} have been obtained using faulty implementations of the Petz recovery circuit, taking into account gate noise on both the \textsc{cnot} as well as the single-qubit gates.

\begin{figure}[t]
\centering
\includegraphics[width=1\columnwidth]{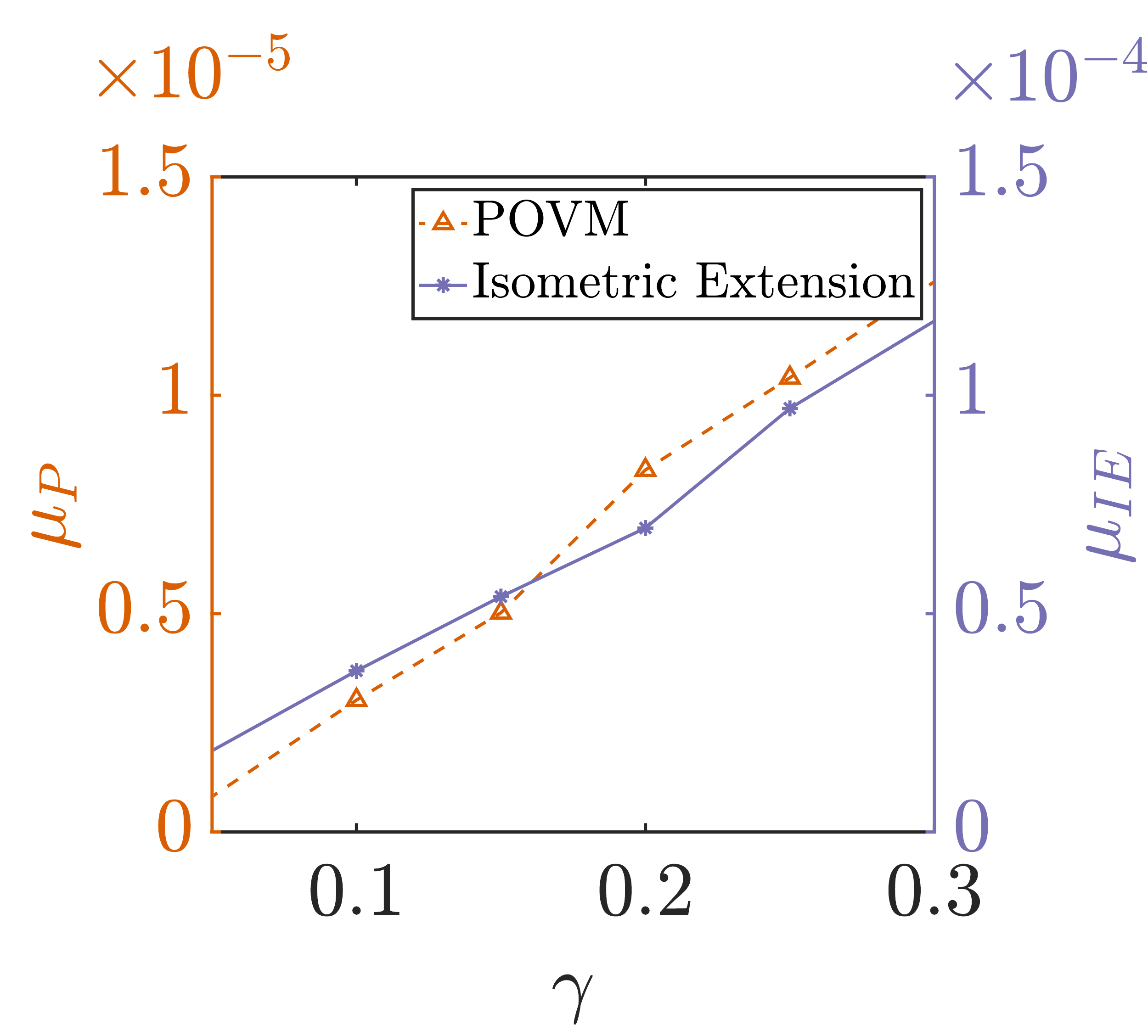}
    \caption{Threshold values $\mu_{IE}$ and $\mu_{P}$ of the depolarizing noise strength on the \textsc{cnot} and single-qubit gates for the Isometric Extension and POVM circuits, respectively.}
    \label{fig:threshold_gamma}
\end{figure}

We perform the noisy simulation in this case, by injecting gate-noise in the qiskit simulator. We assume that the gate-noise can be modelled as a depolarizing channel with noise parameter $\mu$. We expect that the performance of our Petz recovery circuits will deteriorate with increasing value of $\mu$ and would like to identify a \emph{threshold} value of the noise strength upto which our error correction procedure using the Petz recovery circuit succeeds. 

Since the recovery circuits are too complex to analyse, we estimate the gate-noise threshold via noisy simulations. For a given value of damping strength $\gamma$, we define the threshold to simply be the value of noise strength $\mu$ above which the error-corrected qubit has a fidelity smaller than that of the unencoded qubit. Note that this is simply a bound on the noise strength of the physical gates, quite different from a fault-tolerance threshold which comes out of analysing faults in encoded gates.

In Fig.~\ref{fig:threshold_gamma}, we plot the threshold values of the noise strength $\mu$ for two of the Petz recovery circuits, namely, the isometric extension circuit and the POVM-based circuit, as a function of the amplitude-damping strength $\gamma$. The results of our noisy simulations show that the threshold value increases with the increase of the decay strength $\gamma$. For $\gamma\sim0.3$, the POVM-based Petz recovery circuit has a threshold value of the gate noise given by $\mu_P\sim 10^{-5}$ for $\gamma=0.3$. The isometric extension circuit is more tolerant than the POVM-based circuit, with a threshold value $\mu_P\sim 10^{-4}$ for $\gamma \sim0.3$. We attribute the low values of the gate-noise threshold to the fact that the Petz recovery circuits obtained here are not optimized for noisy gates, nor are they designed for fault-tolerant computation. Since the $\gamma$ values range upto $0.3$ in the fidelity maps shown in Fig.~\ref{fig:fidmap}, the noise strengths for the gate-noise have been chosen to lie below the corresponding threshold values in each implementation.


\providecommand{\noopsort}[1]{}\providecommand{\singleletter}[1]{#1}%
\begin{thebibliography}{10}

\bibitem{preskill2018}
John Preskill.
\newblock Quantum {C}omputing in the {NISQ} era and beyond.
\newblock {\em {Quantum}}, 2:79, August 2018.

\bibitem{terhal_qec}
Barbara~M. Terhal.
\newblock Quantum error correction for quantum memories.
\newblock {\em Rev. Mod. Phys.}, 87:307--346, Apr 2015.

\bibitem{tomita2014}
Yu~Tomita and Krysta~M. Svore.
\newblock Low-distance surface codes under realistic quantum noise.
\newblock {\em Phys. Rev. A}, 90:062320, Dec 2014.

\bibitem{leung}
Debbie~W. Leung, M.~A. Nielsen, Isaac~L. Chuang, and Yoshihisa Yamamoto.
\newblock Approximate quantum error correction can lead to better codes.
\newblock {\em Physical Review A}, 56:2567--2573, Oct 1997.

\bibitem{hkn_pm2010}
Hui~Khoon Ng and Prabha Mandayam.
\newblock Simple approach to approximate quantum error correction based on the
  transpose channel.
\newblock {\em Phys. Rev. A}, 81:062342, Jun 2010.

\bibitem{jayashankar2020finding}
Akshaya Jayashankar, Anjala~M. Babu, Hui~Khoon Ng, and Prabha Mandayam.
\newblock Finding good quantum codes using the cartan form.
\newblock {\em Phys. Rev. A}, 101:042307, Apr 2020.

\bibitem{jayashankar2022}
Akshaya Jayashankar and Prabha Mandayam.
\newblock Quantum error correction: Noise-adapted techniques and applications.
\newblock {\em Journal of the Indian Institute of Science}, pages 1--16, 2022.

\bibitem{nielsen}
Michael~A. Nielsen and Isaac~L. Chuang.
\newblock {\em Quantum Computation and Quantum Information}.
\newblock Cambridge University Press, 2000.

\bibitem{fletcher2008}
Andrew~S Fletcher, Peter~W Shor, and Moe~Z Win.
\newblock Channel-adapted quantum error correction for the amplitude damping
  channel.
\newblock {\em IEEE Transactions on Information Theory}, 54(12):5705--5718,
  2008.

\bibitem{jayashankar_ft}
Akshaya Jayashankar, My~Duy~Hoang Long, Hui~Khoon Ng, and Prabha Mandayam.
\newblock Achieving fault tolerance against amplitude-damping noise.
\newblock {\em Phys. Rev. Res.}, 4:023034, Apr 2022.

\bibitem{petz2003}
D{\'e}nes Petz.
\newblock Monotonicity of quantum relative entropy revisited.
\newblock {\em Reviews in Mathematical Physics}, 15(01):79--91, 2003.

\bibitem{barnum2002}
Howard Barnum and Emanuel Knill.
\newblock Reversing quantum dynamics with near-optimal quantum and classical
  fidelity.
\newblock {\em Journal of Mathematical Physics}, 43(5):2097--2106, 2002.

\bibitem{mandayam2012}
Prabha Mandayam and Hui~Khoon Ng.
\newblock Towards a unified framework for approximate quantum error correction.
\newblock {\em Phys. Rev. A}, 86:012335, Jul 2012.

\bibitem{junge2018}
Marius Junge, Renato Renner, David Sutter, Mark~M Wilde, and Andreas Winter.
\newblock Universal recovery maps and approximate sufficiency of quantum
  relative entropy.
\newblock In {\em Annales Henri Poincar{\'e}}. Springer, 2018.

\bibitem{chen2020}
Chi-Fang Chen, Geoffrey Penington, and Grant Salton.
\newblock Entanglement wedge reconstruction using the petz map.
\newblock {\em Journal of High Energy Physics}, 2020(1):1--14, 2020.

\bibitem{gilyen2022_petz}
Andr\'as Gily\'en, Seth Lloyd, Iman Marvian, Yihui Quek, and Mark~M. Wilde.
\newblock Quantum algorithm for petz recovery channels and pretty good
  measurements.
\newblock {\em Phys. Rev. Lett.}, 128:220502, Jun 2022.

\bibitem{lloyd2001engineering}
Seth Lloyd and Lorenza Viola.
\newblock Engineering quantum dynamics.
\newblock {\em Physical Review A}, 65(1):010101, 2001.

\bibitem{terashima2005nonunitary}
Hiroaki Terashima and Masahito Ueda.
\newblock Nonunitary quantum circuit.
\newblock {\em International Journal of Quantum Information}, 3(04):633--647,
  2005.

\bibitem{Gily_n_2019}
András Gilyén, Yuan Su, Guang~Hao Low, and Nathan Wiebe.
\newblock Quantum singular value transformation and beyond: exponential
  improvements for quantum matrix arithmetics.
\newblock {\em Proceedings of the 51st Annual ACM SIGACT Symposium on Theory of
  Computing}, Jun 2019.

\bibitem{low2019}
Guang~Hao Low and Isaac~L Chuang.
\newblock Hamiltonian simulation by qubitization.
\newblock {\em Quantum}, 3:163, 2019.

\bibitem{berry2014_amp}
Dominic~W Berry, Andrew~M Childs, Richard Cleve, Robin Kothari, and Rolando~D
  Somma.
\newblock Exponential improvement in precision for simulating sparse
  hamiltonians.
\newblock In {\em Proceedings of the forty-sixth annual ACM symposium on Theory
  of computing}, pages 283--292, 2014.

\bibitem{Note1}
A two-level unitary is one that acts nontrivially only on a two-dimensional
  subspace of the $d$-dimensional space. We refer to~\cite {nielsen} for
  further details.

\bibitem{cerezo2020}
Marco Cerezo, Alexander Poremba, Lukasz Cincio, and Patrick~J Coles.
\newblock Variational quantum fidelity estimation.
\newblock {\em Quantum}, 4:248, 2020.

\bibitem{ag_2020_fid}
Andr{\'a}s {Gily{\'e}n} and Alexander {Poremba}.
\newblock {Improved Quantum Algorithms for Fidelity Estimation}.
\newblock {\em arXiv:2203.15993}, 2022.

\bibitem{wang_fid}
Qisheng Wang, Zhicheng Zhang, Kean Chen, Ji~Guan, Wang Fang, Junyi Liu, and
  Mingsheng Ying.
\newblock Quantum algorithm for fidelity estimation.
\newblock {\em IEEE Transactions on Information Theory}, 69(1):273--282, 2023.

\bibitem{yordanov2019implementation}
Yordan~S. Yordanov and Crispin H.~W. Barnes.
\newblock Implementation of a general single-qubit positive operator-valued
  measure on a circuit-based quantum computer.
\newblock {\em Phys. Rev. A}, 100:062317, Dec 2019.

\bibitem{andersson2008}
Erika Andersson and Daniel K.~L. Oi.
\newblock Binary search trees for generalized measurements.
\newblock {\em Phys. Rev. A}, 77:052104, May 2008.

\bibitem{gupta2020optimal}
Pragati Gupta and CM~Chandrashekar.
\newblock Optimal quantum simulation of open quantum systems.
\newblock {\em arXiv preprint arXiv:2012.07540}, 2020.

\bibitem{shen2017}
Chao Shen, Kyungjoo Noh, Victor~V. Albert, Stefan Krastanov, M.~H. Devoret,
  R.~J. Schoelkopf, S.~M. Girvin, and Liang Jiang.
\newblock Quantum channel construction with circuit quantum electrodynamics.
\newblock {\em Phys. Rev. B}, 95:134501, Apr 2017.

\bibitem{2011arXiv1106.1445W}
Mark~M. {Wilde}.
\newblock {From Classical to Quantum Shannon Theory}.
\newblock {\em Preface to the Second Edition. In Quantum Information Theory
  (pp. Xi-Xii)}, 2011.

\bibitem{chirolli2008}
Luca Chirolli and Guido Burkard.
\newblock Decoherence in solid-state qubits.
\newblock {\em Advances in Physics}, 57(3):225--285, 2008.

\bibitem{jayashankar2022adaptive}
Akshaya Jayashankar.
\newblock Adaptive codes: constructions, applications and fault tolerance.
\newblock {\em arXiv preprint arXiv:2203.03247}, 2022.

\bibitem{Hyukjoon}
Hyukjoon Kwon, Rick Mukherjee, and M.~S. Kim.
\newblock Reversing lindblad dynamics via continuous petz recovery map.
\newblock {\em Phys. Rev. Lett.}, 128:020403, Jan 2022.

\bibitem{Lindblad_circ}
Richard {Cleve} and Chunhao {Wang}.
\newblock {Efficient Quantum Algorithms for Simulating Lindblad Evolution}.
\newblock {\em arXiv e-prints}, page arXiv:1612.09512, December 2016.

\bibitem{bl_enc}
Lin {Lin}.
\newblock {Lecture Notes on Quantum Algorithms for Scientific Computation}.
\newblock {\em arXiv:2201.08309}, 2022.

\end{thebibliography}
\end{document}